\documentclass[12pt]{article}
\usepackage{url} 

\usepackage{latexsym, amssymb, amscd, amsthm, amsxtra, amsmath,amsthm, mathtools, mathrsfs}
\usepackage{graphics, graphicx, color}
\usepackage{natbib}
\usepackage{ifpdf}
\usepackage[format=hang,indention=-1cm,small]{caption}
\usepackage{subcaption}
\usepackage{multirow}
\usepackage{kotex}
\usepackage{hyperref}
\usepackage{epsfig}
\usepackage{algorithm2e}
\usepackage{arydshln}
\usepackage{ulem}
\usepackage{booktabs}

\newcommand{\blind}{0}

\newtheorem{thm}{Theorem}
\newtheorem{cor}[thm]{Corollary}
\newtheorem{lem}[thm]{Lemma}

\theoremstyle{definition}
\newtheorem{defn}{Definition}
\newtheorem*{defn*}{Definition}

\theoremstyle{remark}
\newtheorem{remark}{Remark}
\newtheorem*{remark*}{Remark}

\newtheorem{example}{Example}

\begin{document}

\def\spacingset#1{\renewcommand{\baselinestretch}%
{#1}\small\normalsize} \spacingset{1}

\if0\blind
{
  \title{\bf Variable selection and basis learning for ordinal classification}
  \author{Minwoo Kim\\
    Department of Statistics, Seoul National University\\
    Sangil Han \\
    Department of Statistics, Seoul National University\\
	Jeongyoun Ahn$^*$\\
	Department of Industrial and Systems Engineering, KAIST \\
	and \\
	Sungkyu Jung\thanks{
    This work was supported by the National Research Foundation of Korea (NRF) grants funded by the Korea government (MSIT) (No. 
      2021R1A2C1093526, 
      2022M3J6A1063021, 
      RS-2023-00218231, 
      RS-2023-00301976, 
      RS-2024-00333399
      ).
	}\hspace{.2cm}\\
	Department of Statistics and Institute for Data Innovation in Science, \\Seoul National University}
  \maketitle
} \fi

\if1\blind
{
  \bigskip
  \bigskip
  \bigskip
  \begin{center}
    {\LARGE\bf Title}
\end{center}
  \medskip
} \fi

\bigskip
\vspace{-15pt}
\begin{abstract}
 We propose a method for variable selection and basis learning for high-dimensional classification with ordinal responses. The proposed method extends sparse multiclass linear discriminant analysis, with the aim of identifying not only the variables relevant to discrimination but also the variables that are order-concordant with the responses. For this purpose, we compute for each variable an ordinal weight, where larger weights are given to variables with ordered group-means, and penalize the variables with smaller weights more severely. A two-step construction for ordinal weights is developed, and we show that the ordinal weights correctly separate ordinal variables from non-ordinal variables with high probability. The resulting sparse ordinal basis learning method is shown to consistently select either the discriminant variables or the ordinal and discriminant variables, depending on the choice of a tunable parameter. Such asymptotic guarantees are given under a high-dimensional asymptotic regime where the dimension grows much faster than the sample size. We also discuss a two-step procedure of post-screening ordinal variables among the selected discriminant variables.  Simulated and real data analyses confirm that the proposed basis learning provides sparse and interpretable basis, as it mostly consists of ordinal variables. 
\end{abstract}

\noindent%
\spacingset{1.5} 

\section{Introduction} \label{sec: intro}

Classification is one of the most important statistical analyses for handling real-world problems. We focus on ordinal classification problems where the responses, or the class labels, are ordered. As an instance, Tumor-Node-Metastasis (TNM) classification describes the severity of tumor progress on the ordinal scale: `Stage 0' $ \prec $ `Stage I' $ \prec $  `Stage II' $ \prec $ `Stage III' $ \prec $ `Stage IV'. Such ordinal responses frequently arise in a variety of application areas, including medical research, social science, and marketing.

In multi-category classification, especially the ultra-high dimensional problems (i.e., the number of predictors is much larger than the sample size), one seeks to reduce the dimension of the problem by utilizing a low-dimensional discriminant subspace, preferably involving only a small number of variables \citep{fan2008high, witten2011penalized,gaynanova2016sparseLDA, ahn2021trace}.
Accordingly, it has been of interest to correctly select variables relevant to discrimination, called \textit{discriminant} variables, and to estimate the basis of the low-dimensional discriminant subspace  for such ultra-high dimensional classification problems.
Under a common-variance model where observations from the $g$th class ($g = 1,\ldots, K$) have a mean $ \mu_g \in \mathbb{R}^p $ and a variance $\Sigma$ (where $ p $ is the number of variables), it amounts to assume that the rows of discriminant basis $\Sigma^{-1}M$, are mostly zero (or, simply \textit{sparse}), and that the non-zero rows correspond to discriminant variables. 
Here, $ M $ can be any mean difference matrix, for example, $M = [\mu_2 - \mu_1,\ldots,\mu_K - \mu_1]$.
There have been a number of important contributions in the literature, including \cite{clemmensen2011sparse}, \cite{gaynanova2016sparseLDA}, \cite{mai2019multiclass}, and \cite{jung2019poi}, for sparse estimation of the discriminant subspace. 
However, these previous works assume that the categories of the response are nominal.

In this work, we study variable selection and basis learning for ordinal classification. With the ordinal responses, it is natural to prefer order-concordant variables with the response, which in turn provides better interpretability to analysts. For instance, in medical research, one of the main concerns is identifying biomarkers that are up-regulated or down-regulated for the progression of a particular disease.
        
If a variable's class-wise means $\mu^1,\ldots,\mu^K$ are either increasing ($\mu^1<\cdots< \mu^K$) or decreasing, we say the variable is order-concordant and call it an \textit{ordinal} variable. However, not all ordinal variables are relevant to discrimination. To deal with this situation, we aim to select the variables that are both ordinal and discriminative. Under the common variance model, we show that the index set of \textit{ordinal discriminant} variables, $J_{disc}^{ord}$, is generally different from the index set of ordinal variables, $J_{ord}$, and also from the index set of discriminant variables, $J_{disc}$. (These index sets are formally defined in Section 2.2) Accordingly, we propose a method of sparse ordinal basis learning that can be tuned to select either the discrimination variables in $J_{disc}$ or the ordinal discriminant variables in $J^{ord}_{disc}$. 

Our proposal uses the sparse multiclass linear discriminant analysis (LDA) framework, studied in \cite{gaynanova2016sparseLDA}, \cite{mai2019multiclass},
and \cite{jung2019poi}. These sparse multiclass LDA methods result  in a row-wise sparse discriminant basis matrix $\widehat{Z}_\lambda \in \mathbb{R}^{p \times (K-1)}$, formulated as the minimizer of the optimization problem
\begin{equation} \label{intro: sparseLDA}
\min_{Z \in \mathbb{R}^{p \times (K-1)}}~\mbox{trace}\left(
	\frac{1}{2} Z^\top\widehat \Sigma Z - Z^\top \widehat M
	\right)
	+ \sum_{j=1}^{p}\lambda_j  \|\widetilde Z_j\|_2,
\end{equation}
where $ \widetilde Z_j $ is the $ j $th row of the candidate matrix $ Z $, and the tuning parameter $\lambda_j$ was assumed to be equal, i.e., $\lambda_j = \lambda$. 
Different authors choose different estimators $ (\widehat \Sigma, \widehat M) $, as discussed in  Section~\ref{sec:2.1}. 
To incorporate the ordinal information, we consider assigning a weight $w_j \in [0,1]$ to each variable, indicating the degrees to which the variable is order-concordant. Ideally, the weight is 1 if the variable is ordinal, and 0 if the variable is either non-ordinal or has no mean difference at all. We devise a two-step procedure for computing ordinal weights by thresholding two types of sample Kendall's $\tau$, which is guaranteed to separate ordinal variables from the others with high probability. By introducing an additional tuning parameter $\eta \ge 1$, we propose to set the penalty coefficients of (\ref{intro: sparseLDA}) by $\lambda_j = \lambda \eta^{1-w_j}.$
While $\lambda \ge 0$ controls the overall sparsity of the resulting basis, larger choices of $\eta$ result in more severe penalization for non-ordinal variables whose weights are $w_j \approx 0$.
Thus, with a suitable choice of $ (\lambda, \eta) $, the proposed basis learning method aims to find both sparse and ordinal discriminant subspace, and is called \textit{sparse ordinal basis learning} or SOBL for short.
 
Under the common-variance sub-Gaussian model, we theoretically show that the SOBL only selects the ordinal discriminant variables in $J_{disc}^{ord}$ with high probability if the tuning parameter $\eta$ is set sufficiently large. On the other hand, for small enough $\eta$, the SOBL selects the discriminant variables in $J_{disc}$ and consistently estimates the discriminant basis, $\Sigma^{-1}M$. When the SOBL is tuned to properly choose the ordinal discriminant variables, it consistently estimates the discriminant subspace corresponding to the variables in $J_{disc}^{ord}$. These asymptotic results are obtained under a high-dimensional setting where the dimension $p$ may grow as fast as any polynomial order of the sample size. Also, numerical studies in Section \ref{sec: numerical study} show that it is possible that SOBL has similar or better classification performance than sparse LDA while producing a sparser basis supported in $J_{disc}^{ord}$.

In addition, we investigate an \textit{ordinality-screened basis learning} (OSBL) procedure, which first fits the sparse LDA basis $ \widehat{Z}_{\lambda} $ and then sets the rows of $\widehat{Z}_{\lambda}$ corresponding to non-ordinal variables as $ 0 $. The OSBL also has consistent variable selection and basis learning properties. Furthermore, under a more strict Gaussian model, the OSBL theoretically outperforms SOBL, yielding improved asymptotic results. Despite theoretical advantages, our empirical studies in Section \ref{sec: numerical study} reveal that SOBL is often preferable to OSBL.
  
There have been a number of important contributions on ordinal classification. The ordinal logistic regression and its variations, such as  cumulative logit model and continuation ratio model \citep{mccullagh1980regression,archer2014ordinalgmifs}, are most notable. To deal with high-dimensional data, \cite{archer2014ordinalgmifs} and  \cite{wurm2021regularized} proposed various regularization approaches to ordinal logistic regression. From a Bayesian perspective, \cite{zhang2018predicting} used a Cauchy prior to hierarchical ordinal logistic models. 
Recently, \cite{ma2021bioinfo} proposed to learn ordinal discriminant basis, using (\ref{intro: sparseLDA}) but with the within-group variance matrix $\widehat\Sigma$ replaced by $\widehat\Sigma + \alpha \mbox{diag}(1-w_1,\ldots, 1-w_p)$, where $\alpha>0$ is a tuning parameter and $w_j$'s are the ordinal weights. This adjustment effectively reduces the (empirical) variance of ordinal variables, resulting in more pronounced selection of ordinal variables. We note that many of these previous work  lack theoretical guarantees. 
We also note that most support vector machine based ordinal classifiers
\citep{shashua2002ranking,chu2005new,qiao2017noncrossing} 
typically perform poorly for high-dimensional data. We numerically compare our proposal with a penalized logistic regression of \cite{archer2014ordinalgmifs} and the basis learning method of \cite{ma2021bioinfo} in Section~\ref{sec: numerical study}.

The rest of the paper is organized as follows. In Section \ref{sec: method}, we formally define the set of  ordinal discriminant variables, and propose the SOBL framework and the OSBL procedure. In Section \ref{sec: ordinal feature weight}, we devise a two-step procedure for computing ordinal weights, and provide theoretical and empirical evidence on the efficacy of the proposed weights. In Section \ref{sec: theory}, the variable selection and basis learning performances of the SOBL are evaluated theoretically. There, we provide non-asymptotic probability bounds, which translate into high-dimensional asymptotic guarantees for sparse ordinal basis estimation. In Section \ref{sec: numerical study}, we numerically show that the proposed methods select a much more sparse basis than competing methods and excel at selecting the ordinal discriminant variables in  simulated and real data examples. Technical details and proofs are contained in the supplementary material.

\section{Sparse ordinal basis learning} \label{sec: method}
We consider a setting of ordinal multiclass linear discriminant analysis with a predictor $X \in  \mathbb{R}^p $ and class membership $Y$, labeled as $1,2,\ldots, K$.
Here, the class labels have natural order relationship of
$ 1 \prec 2 \prec \dots \prec K $.
Given $ Y = g $, we assume $ \mathbb E[X~\vert~Y=g] = \mu_g$ and $\mbox{cov}(X~\vert~Y=g) = \Sigma_w, $ where $\mu_g = (\mu^g_1, \dots, \mu^g_p)^{\top} \in \mathbb{R}^p$ is the mean vector of the $g$th class. The within covariance matrix $ \Sigma_w $ is common across different classes (or groups).
For the class probability, we set $ \pi_g =  \Pr(Y = g)$.
We model that $ X $ conditionally follows a sub-gaussian distribution \citep{vershynin2018high} on the given class; there exists a $ \sigma \in (0, \infty) $ such that for any class $ g $ and $t > 0$,%
\begin{equation} \label{eq:subG-model}
	\sup_{\|v\|_2 = 1} \Pr(|v^{\top}(X - \mu_g)| \ge t ~|~ Y=g) \le 2\exp\left(-\frac{ct^2}{\sigma^2}\right)
\end{equation}
where $ c > 0 $ is an absolute constant.
Note that the model \eqref{eq:subG-model} contains the ordinary Gaussian model $X \sim N(\mu_g, \Sigma_w)$ given $Y = g$ which is generally assumed for linear discriminant analysis setting.
Suppose we have total $ N $ independent observations $ \{(X_i, Y_i)\}_{i=1}^N $.
Let
$
\mathbf{X} = [X_1, \dots, X_N]^\top = [\mathbf{X}_1, \dots, \mathbf{X}_p] \in \mathbb{R}^{N \times p}
$
be the matrix of predictors, where
$ \mathbf{X}_j = (X_{1j}, \dots, X_{Nj})^\top \in \mathbb{R}^N $ collects the $ N $ observations corresponding to the $ j $th variable,
and $ \mathbf{Y} = [Y_1, \dots, Y_N]^{\top} $ be the vector of responses.
For each $ g = 1, \dots, K $, we write
$ n_g = \sum_{i=1}^{N} I(Y_i = g)$ so that $ N = \sum_{g=1}^K n_g $.
We denote the overall mean as $ \mu = \sum_{g=1}^G \pi_g\mu_g $.
The total sample mean is $ \hat \mu = \frac{1}{N}\sum_{i=1}^{N} X_i $ and
the group-wise sample mean for the $g$th group is $ \hat \mu_g = \frac{1}{n_g}\sum_{Y_i = g} X_i $. 

\subsection{Sparse multiclass LDA}\label{sec:2.1}

Multiclass LDA is a widely used projection-based classification method. It finds a low-dimensional subspace, or the LDA subspace, that
maximizes the between-group variance and minimizes the within-group variance simultaneously.
Denote $ \Sigma_b $ for the population between-group covariance matrix, 
$
	\Sigma_b = \sum_{g=1}^K \pi_g (\mu_g - \mu)(\mu_g - \mu)^{\top}.
$

On the population level, LDA seeks the \textit{discriminant vectors}
$\beta_1, \dots, \beta_{K-1}$
which are the solutions to the following sequential optimization problem:
For $ k = 1, 2, \ldots, K-1$,
\begin{equation}
	\max_{\beta_k \in \mathbb{R}^p} \beta_k^\top \Sigma_b \beta_k
	\quad\mbox{ subject to }~
	\beta_k^\top \Sigma_w \beta_k = 1,~
	\beta_k^\top \Sigma_w \beta_l = 0 \mbox{ for } l < k.
\end{equation}
Let $ B = [\beta_1, \dots, \beta_{K-1}] \in \mathbb{R}^{p \times (K-1)}$ be the matrix of discriminant vectors. We call the subspace $\mathcal C(B) $ as the \textit{population discriminant subspace} where $ \mathcal C(\cdot) $ stands for the column space of the given matrix.
It is known that for any basis matrix $ M \in \mathbb{R}^{p \times (K-1)} $ of $ \mathcal C ( \Sigma_b) $, $\mathcal C(B) = \mathcal{C}(\Sigma_w^{-1}M) = \mathcal C (\Sigma_T^{-1}M)$,
where $ \Sigma_T = \mbox{cov}(X) = \Sigma_w + \Sigma_b $  \citep{safo2016general,gaynanova2016sparseLDA}.
In this perspective, estimating the \textit{population discriminant basis} $ \Psi = \Sigma^{-1}M $, for any $ \Sigma \in \{\Sigma_w, \Sigma_T\}$,  instead of directly estimating discriminant vectors $ B $, is enough for the discriminant subspace estimation.
The natural way to estimating $ \Psi $ is substituting $ \Sigma $ and $ M $ with their sample versions.
For the sample versions of covariances, we take
$ \widehat \Sigma_w = \frac{1}{N}\sum_{g=1}^{K}\sum_{Y_i =g }	(X_i - \hat \mu_g)(X_i - \hat \mu_g)^\top,$ and 
$ \widehat{\Sigma}_T = \widehat \Sigma_w + \widehat \Sigma_b $, where
$ \widehat \Sigma_b = \sum_{g=1}^{K} \frac{n_g}{N} (\hat \mu_g - \hat \mu)(\hat \mu_g - \hat \mu)^\top. $

However, in high-dimensional situations with $ p > N $, $ \widehat{\Sigma}_w $ and $ \widehat{\Sigma}_T $  are singular and classification rules based on such a sample subspace perform poorly \citep{bickel2004some}. Moreover, since the classifier involves so many variables, it is not easy to interpret. 
To resolve these problems, sparse estimators of the discriminant subspace have been proposed. In particular, the sparse multiclass LDA estimators proposed by \cite{gaynanova2016sparseLDA}, \cite{mai2019multiclass} and \cite{jung2019poi} are all formulated by 
the same form of convex optimization problems.
To motivate these estimators, observe that
\begin{equation} \label{eq: pop-basis-lowdim}
	\Psi
	= \underset{Z \in \mathbb{R}^{p \times (K-1)}}{\arg\min}
	\mbox{trace}\left(
	\frac{1}{2} Z^\top \Sigma Z - Z^\top  M
	\right).
\end{equation}
Sparse multiclass LDA methods estimate a sparse basis of discriminant subspace by adding a sparsity-inducing convex penalty term to the sample version of (\ref{eq: pop-basis-lowdim}):
\begin{equation} \label{eq: sparse LDA}
	\widehat Z_\lambda = \underset{Z \in \mathbb{R}^{p \times (K-1)}}{\arg\min}
	\left\{
	\mbox{trace}\left(
	\frac{1}{2} Z^\top \widehat \Sigma Z - Z^\top \widehat M
	\right)
	+ \lambda \sum_{j=1}^{p} \|\widetilde Z_j\|_2
	\right\},	
\end{equation}
where $ \widetilde Z_j $ is the $j$th row of $Z$ and
$ \lambda > 0 $ is a sparsity-inducing tuning parameter.
The group-lasso \citep{yuan2006grouplasso} penalty term $\lambda \sum_{j=1}^{p} \|\widetilde Z_j\|_2$
in (\ref{eq: sparse LDA}) 
leads to variable-wise sparsity of the estimated basis since the elements of $Z$ are grouped row-wise. 
We note that the variable-wise sparsity is preserved under rotation, which is a desirable property in basis learning. See \cite{jung2019poi} for a detailed discussion on other choices of penalty.

\cite{gaynanova2016sparseLDA}, \cite{mai2019multiclass} and \cite{jung2019poi} choose to use different combinations of $ (\widehat \Sigma, \widehat M) $ in obtaining their sparse LDA basis matrix $ \widehat Z_{\lambda} $.
\cite{gaynanova2016sparseLDA} uses $\widehat \Sigma = \widehat \Sigma_T$ and $ \widehat M $ given by setting the $r$th column of $\widehat M$ to be
$
	\hat m_r = \frac
	{\sqrt{n_{r+1}} \sum_{i=1}^{r}n_i(\hat \mu_i - \hat \mu_{r+1})}
	{\sqrt{N\sum_{i=1}^{r}n_i \sum_{i=1}^{r+1}n_i}}
$
for $r = 1, \dots, K-1$.
This choice of $ (\widehat \Sigma, \widehat M) $ will be referred to as
multi-group sparse discriminant analysis (MGSDA), following the terminology of \cite{gaynanova2016sparseLDA}.
The combination of $ \widehat \Sigma = \widehat \Sigma_w$ and
$ \widehat M = [\hat \mu_2 - \hat \mu_1, \dots, \hat \mu_K - \hat \mu_1] $
are used in \cite{mai2019multiclass}, which will be called the
multiclass sparse discriminant analysis (MSDA).
Finally, \cite{jung2019poi} chose $ \widehat \Sigma = \widehat \Sigma_w$ and to use
the leading $ K-1 $ eigenvectors of $ \widehat \Sigma_b $
corresponding to nonzero eigenvalues of $ \widehat \Sigma_b $ for the columns of $ \widehat M $.
Following \cite{jung2019poi}, this choice will be called
fast penalized orthogonal iteration (fastPOI). 

\subsection{Ordinal discriminant variables}\label{sec:2.2}

We categorize the $p$ predictors, or variables, by the roles they play in the classification under the LDA model.
Let $ J =  \{1, 2, \dots, p\} $ be the index set of all variables.
For $ j \in J $, the $ j $th variable is a \textit{mean difference} variable
if $ \mu_j^{g_1} \neq \mu_j^{g_2} $ for some $ g_1 \neq g_2 $,
and $ J_{md} $ denotes the index set of all mean difference variables. We say that
the $ j $th variable is a \textit{noise} variable  if
$ \mu_j^1 = \mu_j^2 = \dots = \mu_j^{K} $, and  $ J_{noise} $ denotes the set of noise variables.
We further split $ J_{md} $ as follows.
The $ j $th variable is an order-concordant mean-difference variable, or \textit{ordinal} variable, if $ j \in J_{md} $ and $\mu^1_j \le \mu^2_j \le \dots \le \mu^K_j $ or $\mu^1_j \ge \mu^2_j \ge \dots \ge \mu^K_j,$
and	we denote the set of ordinal variables by $ J_{ord} $; 
and finally, 
the $ j $th variable is a nominal (non-ordinal mean-difference) variable if $ j \in J_{nom}$ where $ J_{nom} := J_{md} \setminus J_{ord}$.

Next, we define variables relevant to the population discriminant subspace.
The index set of all discriminant variables $ J_{disc} $ is defined as
$J_{disc} = \{j: \Psi_{j, :} \neq \mathbf{0}_{K-1}\} $,
where $\Psi_{j, :}$ denotes the $ j $th row of $ \Psi $.
Note that  $ J_{disc} $ consists of all variables contributing to the Bayes rule classification under the Gaussian model.

In general, $ J_{md} $ is not the same as $J_{disc}$, and furthermore, it is possible to be either $J_{md} \not\subset J_{disc} $ or $J_{disc} \not\subset J_{md} $ as shown in Proposition 1 and following examples of \cite{mai2012direct}.
Thus, only using  the mean-difference variables in classification, as done in e.g. ``independence rules'' \citep{bickel2004some, fan2008high}, may fail to capture the discriminant variables.
Similarly, in ordinal classification problems, not all ordinal variables contribute to discrimination.
To elaborate on this point, we define ordinal discriminant variables as follows.
\begin{defn} \label{def: odr-disc-var}
	For $ j \in \{1, 2, \dots, p\} $,
	the $ j $th variable is an \textit{ordinal discriminant variable}
	if $ j \in J_{disc}^{ord} \coloneqq J_{ord} \cap J_{disc} $.
\end{defn}

In general, $J_{disc}^{ord} \neq J_{disc}$ and $J_{disc}^{ord} \neq J_{ord}$.

\begin{example} \label{example-1}
	Let us consider $ p=8 $, $ K=3 $,
	$ \Sigma_w = 0.5 (I_p + \mathbf 1_p \mathbf 1_p^\top)$ and
	\[
	[\mu_1, \mu_2, \mu_3]
	=
	\left[
	\begin{array}{cccc:cccc}
		0.5 & 0   & 0  & 0  & 0 & 0 & 0 & 0 \\
		1   & 0.5 & 1  & -1 & 3 & 2 & -1 & -0.5 \\
		1.5 & 1.0 & 2  & -1.5 & 2 & -0.5 & 2 & 3
	\end{array}
	\right]^\top.
	\]
	Then for the three LDA methods introduced at Section \ref{sec:2.1},
	we have
	\begin{align*}
		\Psi_{\mbox{\scriptsize{MGSDA}}} &= \begin{bmatrix}
			0 & 0 & -0.187 & 0.305 & -0.374 & -0.135 & 0.157 & 0.0480 \\
			0 & 0 & -0.102 & 0.251 & -0.0639 & 0.196 & -0.140 & -0.243	
		\end{bmatrix}^\top, \\
		\Psi_{\mbox{\scriptsize{MSDA}}} &= \begin{bmatrix}
			0 & 0 & 1 & -3 & 5 & 3  & -3 & -2  \\
			0 & 0 & 2 & -5 & 2 &-3  & 2  & 4
		\end{bmatrix}^\top, \\
		\Psi_{\mbox{\scriptsize{fastPOI}}} &= \begin{bmatrix}
			0 & 0 & -0.270 & 0.609 & 0.256 & 0.996 & -0.796 & -1.07 \\
			0 & 0 & 0.320 & -0.909 & 1.12 & 0.505 & -0.556 & -0.236 	
		\end{bmatrix}^\top.
	\end{align*}
	In this example, we have
	$ J_{disc} = \{3, 4, 5, 6, 7, 8\} $, $ J_{ord} = \{1, 2, 3, 4\} $ and
	$ J_{disc}^{ord} = \{3, 4\} $.
	So, the first and second variables are ordinal
	but do not contribute to discriminant directions.
	See Figure \ref{fig: index} for an illustration.
\end{example}

\begin{figure}[t]
	\begin{center}
		\includegraphics[width=0.4\textwidth]{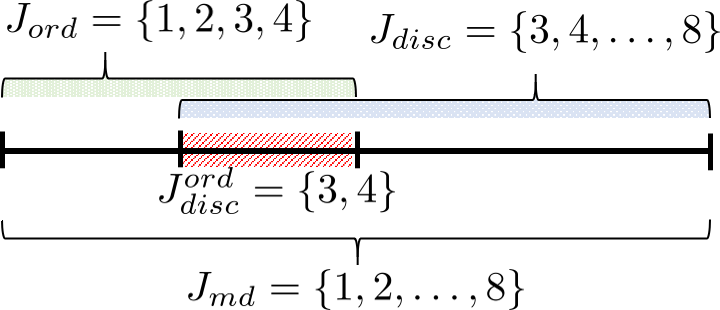}
		\caption{
        Visual illustration of variable sets in Example \ref{example-1} with respect to their variable types.
        }
		\label{fig: index}
	\end{center}
\end{figure}

We aim to find a discriminant basis that contains ordinal information, and consequently, our target variables are the variables in $J_{disc}^{ord}$. Our proposed methods can be tuned to detect either only the ordinal discriminant variables or the discriminant variables, as we demonstrate in later sections, while other methods do not capture $ J_{disc}^{ord} $ properly.

\subsection{Proposed methods} \label{sec:2.3}
We incorporate variable-wise order-concordance information 
into the sparse LDA basis learning framework (\ref{eq: sparse LDA}).
To extract and use ordinal information,
we adopt the idea of ``ordinal feature weights'' in \cite{ma2021bioinfo}.
For each $ j \in J $, we will define and use a variable-wise ordinal weight $w_j \in [0, 1]$, computed from the data $ (\mathbf{X}_j, \mathbf{Y}) $.
The ordinal weight indicates the degrees to which a variable is order-concordant with the class order, i.e., larger weights for order-concordant variables.
A natural candidate for $w_j = w_j(\mathbf{X}_j, \mathbf{Y})$ is the absolute value of rank correlation, as used in \cite{ma2021bioinfo}, but is not an ideal choice. We will return to this matter shortly.
Once the ordinal weights $w_j$ are given, we adjust the penalty coefficient for the  $ j $th variable by
$
\lambda_j = \lambda \eta^{1-w_j},
$
where $ \lambda > 0 $ and $ \eta \ge 1 $ are tuning parameters,
and attach $ \lambda_j $ to the group lasso penalty term of (\ref{eq: sparse LDA})
and propose the following basis:
\begin{equation} \label{eq: SOBL}
	\widehat Z_{\eta, \lambda}^{ord} = \underset{Z \in \mathbb{R}^{p \times (K-1)}}{\arg\min}
	\left\{
	\mbox{trace}\left(
	\frac{1}{2} Z^\top\widehat \Sigma Z - Z^\top \widehat M
	\right)
	+ \lambda \sum_{j=1}^{p} \eta^{1-w_j} \|\widetilde Z_j\|_2
	\right\}.
\end{equation}
A block-coordinate descent algorithm \citep{tseng1993desc} can efficiently solve the convex minimization problem (\ref{eq: SOBL}).
The detailed procedure is given in supplementary material S2. 

When the additional tuning parameter $\eta$ is set to 1, the estimated basis $\widehat Z_{1, \lambda}^{ord}$ is exactly the $\widehat Z_\lambda$ from the sparse multiclass LDA (\ref{eq: sparse LDA}).
On the other hand, when $\eta > 1$ (and $\lambda >0$), the penalty coefficient for the $j$th variable, $\lambda_j = \lambda \eta^{1-w_j}$, is a decreasing function of $w_j$. In particular, smaller ordinal weight $w_j \approx 0$ results in larger $\lambda_j$ which in turn gives more severe penalization for the magnitude of the $j$th variable loadings $\widetilde{Z}_j$ in the basis matrix. 
Thus, the basis obtained from (\ref{eq: SOBL}) tends to include ordinal variables with larger weights and also to discard non-ordinal and noise variables with smaller weights. In this vein, our proposal (\ref{eq: SOBL}) may be called a method for \textit{sparse ordinal basis learning} (SOBL).

Other formulations for the penalty coefficient that provide theoretical properties comparable to those of $\eta^{1-w_j}$ are available. For instance, a penalty coefficient $\lambda_j = \lambda \eta^{(1-w_j)^2}$ would also work. While exploring the best choice for the penalty term may be interesting, we do not address this specific issue in this paper.

The result of SOBL is highly dependent on the ordinal weights $w_j$. Ideally, properly chosen weights separate the variables in $J_{ord}$ from those in $J_{nom} \cup J_{noise}$, by assigning larger weights for $J_{ord}$ and smaller weights for $J_{nom} \cup J_{noise}$.
In Section \ref{sec: ordinal feature weight},
we review some candidates for ordinal weights based on rank correlations and
monotone trend tests, and
propose a two-step ordinal weight that satisfies
$w_j = 1$ if $j \in J_{ord}$, and $w_j = 0$ if $j \in J_{nom} \cup J_{noise}$
with high probability under some regularity conditions.

Denote the selected variables of SOBL by $\widehat D^{sobl}
	:= \widehat{D}(\widehat{Z}_{\eta, \lambda}^{ord}) 
	= \{j: \mbox{the }j\mbox{th row of } \widehat Z^{ord}_{\eta, \lambda} \neq 0\}.$
Taking $ \eta $ large enough, the selected variables in $ \widehat D^{sobl} $ are targeted at estimating $ J_{disc}^{ord} =  J_{disc} \cap J_{ord} $ which is generally not equal to $ J_{disc} $ or $ J_{ord} $ as seen in Example \ref{example-1}. Thus, the variable selection by SOBL is not in general equivalent to either a variable screening (of screening-in $J_{ord}$) or a variable selection by vanilla sparse multiclass LDA, targeted to select $J_{disc}$.
In Section \ref{sec: theory}, we show theoretically that, for large enough $\eta$, $ \widehat D^{sobl} \subset J_{disc}^{ord} $ and for small $\eta$, $ \widehat D^{sobl} \subset J_{disc}$ in both non-asymptotic and high-dimensional asymptotic settings.
We note that these types of theoretic results are completely new to
the ordinal classification literature.
Based on the theoretical result, we propose a tuning parameter selection scheme
in Section \ref{sec: numerical study}.
In the same section, we also numerically demonstrate that, using the proposed tuning procedure, the variables selected by SOBL are mostly ordinal discriminant variables in $J_{disc}^{ord} $, a feature that other variable selection methods do not have.

In addition to the proposed SOBL basis of \eqref{eq: SOBL}, one can consider obtaining a sparse basis as follows: Fit the sparse multiclass LDA \eqref{eq: sparse LDA} and subsequently screen out nominal variables based on ordinal weights. 
For a predetermined threshold $ \theta_w \in [0, 1] $, 
the set of screened ordinal variables is defined as $ \hat{J}_{ord} := \hat{J}_{ord}(\theta_w) = \{j: w_j \ge \theta_w\}, $ and the following ordinal basis $ \widehat{Z}_{\lambda}^{scr} $ is defined as $ (\widehat{Z}_{\lambda}^{scr})_{j,:} = (\widehat{Z}_{\lambda})_{j, :} $ if $ j \in \hat{J}_{ord} $, and $ (\widehat{Z}_{\lambda}^{scr})_{j,:} = 0 $ otherwise.
(Here, the subscript denotes the $ j $th row.)
We call this screening procedure as \textit{ordinality-screened basis learning} (OSBL) and call $\widehat{Z}_{\lambda}^{scr}$ as OSBL basis. 
Thus, the selected variables are in $ \widehat{D}(\widehat{Z}^{scr}_{\lambda}) = \hat{J}_{ord} \cap \hat J_{disc}$, where $\hat J_{disc} = \widehat{D}(\widehat{Z}_{\lambda}) $.
If both $ \hat J_{disc} = J_{disc} $ and $ \hat J_{ord} = J_{ord} $ occur with high probability, the OSBL basis also consistently selects variables contained in $ J_{disc}^{ord}. $
In Section \ref{sec: theory}, we theoretically demonstrate that the OSBL basis has the variable selection consistency under the sub-Gaussian model (\ref{eq:subG-model}), and even has a better asymptotic result than  SOBL under the Gaussian model.
Although both methods consistently select variables in $ J_{disc}^{ord} $, the two methods work differently in finite samples.
SOBL has more flexibility through utilizing the tuning parameter $ \eta $, and our empirical results, presented at Section 5, show that SOBL performs better than the OSBL basis in classification accuracy.

\begin{remark}
It is natural to ask whether SOBL is equivalent to the other choice of two-step variable selection, first screening in ordinal variables and then applying the multiclass LDA. The answer is no.
Specifically, while SOBL aims to select the  ordinal discriminant variables in $J_{disc}^{ord}$, the set $J_{disc}^{ord}$ is not in general equal to the variables selected by such a two-step method.
As an example, consider the model introduced in Example \ref{example-1}. Suppose we only use the variables in $ J_{ord} = \{1, 2, 3, 4\} $ in finding the discriminative subspace. This two-step approach gives the discriminative basis matrix
$ \Psi_{\mbox{\scriptsize{MSDA}}}^o = (\Sigma_{1:4, 1:4})^{-1} M_{1:4, :} = [v_1, v_2]$ where $v_1=(0.6, 0.6, 1.7, -2.4)^\top$ and $v_2=(1, 1, 3, -4)^\top$. 
This implies that resulting
discriminant variables are indexed $ \{1, 2, 3, 4\} $.
This does not coincide with $ J_{disc}^{ord} = \{3, 4\} $, which is the target for SOBL and the OSBL basis.
We mention in passing that such a two-step approach  for ordinal classification was considered in \cite{zhang2018predicting}.
\citeauthor{zhang2018predicting} proposed to first screen variables based on $p$-values obtained from univariate ordinal logistic regression, then to fit a Bayesian hierarchical ordinal logistic model with the screened variables.
\end{remark}

\section{The choice of ordinal weights} \label{sec: ordinal feature weight}
The ordinal weight $w_j$ measures order-concordance between the $N $ sample of the $ j $th variable $ \mathbf{X}_j $ and $ \mathbf{Y} $. We briefly review two natural choices for ordinal weights, stemming from rank correlations and monotone trend tests, respectively, and reveal their drawbacks. We then propose a two-step composite procedure of computing ordinal weights and show that our proposal provides the largest empirical gap between weights in $ J_{ord} $ from those in $ J_{ord}^c $ and is also guaranteed to identify $ J_{ord} $ asymptotically.

\subsection{Rank correlation ordinal weights}
Natural choices for ordinal weights are given by 
Spearman's rank correlation and Kendall's $ \tau $ \citep{kendall1990rank}.
Spearman's rank correlation for the $ j $th variable is defined by the sample Pearson correlation coefficient between two rank variables with respect to $ \mathbf X_j $ and $ \mathbf Y $.
The sample Kendall's $ \tau $ of the $ j $th variable, denoted by $ \hat \tau_j $, is defined as
\begin{equation} \label{eq: tau}
\hat \tau_j(\mathbf X, \mathbf Y) = \frac{1}{N(N-1)/2}
\sum_{1 \le i_1 < i_2 \le N} \mbox{sign}(X_{i_2 j} - X_{i_1 j}) \cdot \mbox{sign}(Y_{i_2}-Y_{i_1}).
\end{equation}
\cite{ma2021bioinfo} used $w_j^R=\vert \mbox{rank corr}(\mathbf{X}_j, \mathbf{Y}) \vert$ as an
ordinal weight of the $ j $th variable.
Here,  `$\mbox{rank corr}$' is either Spearman's rank correlation or Kendall's $ \tau $.
From the definition of rank correlations, higher $ w_j $ indicates that the $ j $th variable has
an ordinal trend along with the group ordinality $ 1 \prec \dots \prec K $.

While for $ j \in J_{noise} $, $ w_j^R \to 0 $ as $ N \to \infty $,
it is not generally true that $ w_j^R \to 1 $ as $ N \to \infty $ for $ j \in J_{ord} $.
Moreover, for $ j \in J_{nom} $, $ w_j^R $ does not necessarily converge to $ 0 $ as $ N \to \infty $. Thus, the gap, $ \min\{w_j^R: j \in J_{ord}\} - \max \{w_j^R: j \in J_{ord}^c\} $, can be narrow.

\subsection{Ordinal weights based on monotone trend tests}
For each $ j \in J $, consider testing
whether the $ j $th variable is ordinal or not, and set
\begin{equation} \label{eq: trend test}
	H_0^{(j)}: j \in J_{ord}^c
	~\mbox{ vs. }~
	H_a^{(j)}: j \in J_{ord}.
\end{equation}
Suppose a monotone trend test is used to test (\ref{eq: trend test}), and let $p_j$ be the $p$-value of the test applied for the $j$th variable.
Define an ordinal weight for the $ j $th variable by $w_j^P = 1 - p_j$.
Although  $ w_j^P $ may not converge to $0$ for all $ j \in  J_{ord}^c$,
for $ j \in J_{ord} $, $ w_j^P \to 1 $ as $ N \to \infty $ if the test procedure is
``well-designed''.
Hence, $ w_j^P$ is a viable candidate for an ordinal weight.

Most existing trend tests,
such as Jonckheere trend test \citep{jonckheere1954distribution}, Cuzick test \citep{cuzick1985wilcoxon}, and likelihood ratio test \citep{robertson1988order},
only consider the null hypothesis of
$ \mu_j^1 = \mu_j^2 = \dots = \mu_j^K $,
which causes a failure of controlling type I error rate when $ j \in J_{nom} $
\citep{robertson1988order, hu2020testing}.
To handle the whole parameter space, \cite{hu2020testing} proposed
sequential test procedures based on the one-sided two sample $ t $-test
for testing
\begin{equation} \label{eq: broad test}
	H_0^{(j)}:~ \mbox{not } H_a^{(j)} \quad \mbox{vs.} \quad
	H_a^{(j)}:~ \mu^{1}_j < \mu^{2}_j < \cdots < \mu^{K}_j  ~\mbox{ or }~
	\mu^{1}_j > \mu^{2}_j > \cdots > \mu^{K}_j.
\end{equation}
The ordinal weight based on sequential $t$-tests is given as follows.  
For each $ j \in J $,
\begin{itemize}
	\item[1.] (Increasing order)
	Calculate the $p$-value $ p_i^{inc} $ of the one-sided $ t $-test for testing
	$ H_{0,i}^{(j)}: \mu^i_j \ge \mu^{i+1}_j ~ \mbox{vs.} ~
	H_{a,i}^{(j)}: \mu^i_j < \mu^{i+1}_j$
	for each $ i = 1, \dots, K-1 $ and define
	$ p^{inc} = \max\{p_1^{inc}, \dots, p_{K-1}^{inc}\} $.

	\item[2.] (Decreasing order)
	Calculate the $p$-value $ p_i^{dec} $ of the one-sided $t$-test for testing
	$ H_{0,i}^{(j)}: \mu^i_j \le \mu^{i+1}_j ~ \mbox{vs.} ~ H_{a,i}^{(j)}:
	\mu^i_j > \mu^{i+1}_j$ for each $ i = 1, \dots, K-1 $ and define
	$ p^{dec} = \max\{p_1^{dec}, \dots, p_{K-1}^{dec}\} $.
	\item[3.] The overall $p$-value is $ p_j = \min(p^{inc}, p^{dec}) $ and 
    the ordinal weight is given by $ w_j^P = 1-p_j $.
\end{itemize}

Although the parameter space of $ H_a^{(j)} $ in (\ref{eq: broad test}) is slightly smaller than that given by $ J_{ord} $ in (\ref{eq: trend test}),
this strict ordinality simplifies the analysis, and makes the ordinality more interpretable.
For sufficiently large $ N $,
$ w_j^P \approx 0 $ for $ j \in J_{nom} $, and $ w_j^P \approx 1 $ for $ j \in J_{ord} $, if $ H_a^{(j)} $ of (\ref{eq: broad test})
is also true, as desired.
However, there is no evidence for $ w_j^P \to 0 $ as $ N \to \infty $ when $ j \in J_{noise} $.

\subsection{Two-step ordinal weight} \label{subsec: two-step}
To construct an ordinal weight $ w_j $ satisfying $ w_j \to 1 $ for $ j \in J_{ord} $ and
$ w_j \to 0 $ for $ j \in J_{ord}^c $,
we propose a two step construction based on  Kendall's rank correlation.

In the first step, we use the variable-wise Kendall's $ \tau $ (\ref{eq: tau})
to screen out the noise variable.
Second, Kendall's $ \tau $ calculated from group means and class labels
is used to discern ordinal variables from nominal variables.
A detailed procedure is now given.
Let $ \theta_1, \theta_2 \in (0, 1) $ be pre-defined constants.
For each $ j \in J $,
\begin{enumerate}
	\item (Screening out $ J_{noise} $) Calculate
	$$ \hat \tau_j =  \frac{2}{N(N-1)}\sum_{1 \le i_1 < i_2 \le N}
	\mbox{sign}(X_{i_2 j} - X_{i_1 j}) \mbox{sign}(Y_{i_2} - Y_{i_1}).$$
	Set $ w_j = 0 $ if $ \vert \hat \tau_j \vert \le \theta_1 $.
	Otherwise, proceed to Step 2.
	\item (Identifying $ J_{ord} $ and $ J_{nom} $) Calculate
	$$ \tilde{\tau}_j = \frac{2}{K(K-1)} \sum_{1 \le g_1 < g_2 \le K} \mbox{sign}
	(\hat \mu_j^{g_2} - \hat \mu_j^{g_1}).$$
	Ordinal weight is given as $ w_j = 1 $ if $ \vert \tilde \tau_j \vert > 1 - \theta_2 $ and
	$ w_j = 0 $ otherwise.
\end{enumerate}
Summing up, the proposed \textit{two-step} ordinal weight for the $ j $th variable is
\begin{equation} \label{eq: theo ordinal weight}
	w_j = I(\vert \hat{\tau}_j \vert > \theta_1) I(\vert \tilde{\tau}_j \vert > 1 - \theta_2).
\end{equation}

To see why the noise variables in $ J_{noise} $ are screened out in the first step,
note that $\hat \tau_j$ satisfies
$\mbox{Pr}\left( \vert \hat \tau_j - \mathbb E \hat \tau_j \vert \ge t	 \right)
\le 2\exp( -Nt^2 / 8)$ for any $ j \in J $
as shown in Lemma  S1 in the supplementary material. 
Since $\hat \tau_j$ has a sub-Gaussian tail bound around $\mathbb E \hat \tau_j$, we shall investigate $\mathbb E \hat \tau_j$ for $j  \in J_{noise}$ and $j \in J_{md}$. 
The expectation of $ \hat \tau_j $ becomes $ 0 $ when $ j \in J_{noise} $. (See Section S1.1 of the supplementary material for derivation.)
In contrast, $| \mathbb E \hat{\tau}_j |$ is strictly greater than zero for all $ j \in J_{md} $.
If $ N $ is taken large enough, then for an appropriately chosen constant $\theta_1$,
we have $0 \approx \hat \tau_j < \theta_1$ for all $ j \in J_{noise} $, and
$| \mathbb E \hat{\tau}_j | \approx \vert \hat \tau_j \vert > \theta_1 > 0$ for
all $ j \in J_{md} $ with high probability.
So, we can detect whether the group mean structure is noise or not by thresholding
$ \hat \tau_j $ as in Step 1, $ I(|\hat \tau_j| > \theta_1) $.
Note that Step 1 does not always discern the ordinal variables from non-ordinal mean-difference variables. This is because $\max_{j \in J_{nom}} |\mathbb E\hat\tau_j|$ can be arbitrarily close to, or even larger than, $\min_{j \in J_{ord}}  |\mathbb E\hat\tau_j|$.

In the second step, the group means are only used in computing  $\tilde \tau_j$.
We assume that the strict ordinality holds for all $ j \in J_{ord} $, i.e.,
$ \mu_j^1 < \mu_j^2 < \dots < \mu_j^K $ or $ \mu_j^1 > \mu_j^2 > \dots > \mu_j^K $.
Since for all $g = 1,\ldots, K$, $ \hat \mu_j^{g} \to \mu_j^{g}$ almost surely as
$ N \to \infty $, we have
for any $j \in J_{ord}$, $ \vert  \tilde \tau_j \vert \to 1$ almost surely (a.s.).  Moreover, $ \lim_{N \to \infty} \vert \tilde \tau_j \vert < 1 \mbox{ a.s.}$
when $ j \in J_{nom} $.
This implies that by thresholding $ \tilde \tau_j $, as in $ I(|\tilde \tau_j| > 1-\theta_2) $,
we can determine whether the given variable is ordinal or not. Note that using Step 2 alone does not screen out the noise variables since, for any $j \in J_{noise}$, $\tilde \tau_j$ can be 1 with positive probability.

\subsection{Maximal separating property of the two-step ordinal weights} \label{subsec: ts-theory}

Theoretically, appropriately chosen threshold values $ \theta_1$ and $ \theta_2 $
achieve a maximal separating property in the sense that
$ w_j = 1 $ for $ j \in J_{ord} $ and $ w_j = 0 $ for $ j \in J_{ord}^c $.
Define
$ \Delta = \min_{j \in J_{md}} \vert \mathbb E \hat \tau_j \vert $ and
$\vartheta_{\min} = \min_{j \in J_{md}} \min_{1 \le g_1 < g_2 \le K }
\vert \mu_{j}^{g_2} - \mu_{j}^{g_1}\vert / \sigma,$
where $\sigma > 0$ is the sub-gaussian parameter appeared in \eqref{eq:subG-model}.
We require the following conditions:
\begin{itemize}
	\item[(C1)] There exist $ c_1, c_2 > 0 $ such that $ c_1 \le \pi_g \le c_2 $ for all $ g = 1, \dots, K $.
	\item[(C2)] For $j \in J_{noise}$, it holds that $X_j - \mu_j \overset{d}{=} -(X_j - \mu_j) ~|~ Y$, where $\mu_j \equiv \mu_j^1 = \dots = \mu^K_j$.
    \item[(C3)] $ \Delta > 0 $.
	\item[(C4)] $ \vartheta_{\min} > 0 $.
\end{itemize}
Condition (C1) guarantees balanced class sizes.
Condition (C2) assumes symmetry for noise variables and implies that $ \mathbb E \hat \tau_j = 0 $ for $j \in J_{noise}.$
Condition (C3) is necessary for separating mean difference variables from noise variables.
We note that (C3) is a loose condition since we can allow $\Delta \to 0$ as $N \to \infty$ in our asymptotic results in Section \ref{subsec: asymptotic}.
Condition (C4) requires all class means are different, i.e.,
$ \mu_j^{g_1} \neq \mu_j^{g_2} $ for all $ j \in J_{md} $ and $ 1 \le g_1 < g_2 \le K $,
which
is necessary for separating the Kendall's $ \tau $
between $ J_{ord} $ and $J_{nom}$ in the second step.
We emphasize again that
we allow $\vartheta_{\min} \to 0$ as $N \to \infty$ in Section \ref{subsec: asymptotic}. 
The following theorem shows that the two-step ordinal weights
are guaranteed to maximize the separation of $ w_j $ values between $ J_{ord} $ and $ J_{ord}^c $
with high probability.
\begin{thm} \label{thm: ordinal weight}
	Assume conditions (C1)-(C4).
	Let $ \{w_j\}_{j=1}^p $ be the two-step ordinal weights with
    $
	\theta_1 =(\Delta/2)\vee \max_{j \in J_{noise}} |\hat \tau_j| $ and $	\theta_2=\frac{2}{K(K-1)}.
	$
    Then, for sufficiently large $N \gtrsim (\vartheta_{\min}^2 \wedge 1)^{-1}$, 
    $w_j=I(j \in J_{ord})$ holds
    for all $ j = 1, \dots, p $ with probability at least
    $
	1 - 2p\exp(-N \Delta^2/32)
	-  6p K^2\exp(-CN(\vartheta_{\min}^2 \land 1)),
	$
	where $ C > 0 $ is a generic constant only depending on $ c_1 $ and $ c_2 $.
\end{thm}

For any $ j \in J_{md} $, $ |\hat \tau_j| > \Delta/2 $ holds with high probability. 
On the other hand, $ |\hat \tau_j| \le \theta_1 $ for any $ j \in J_{noise} $.
Thus, $ \theta_1 $ in Theorem \ref{thm: ordinal weight} successfully filters out noise variables.
For the choice of $ \theta_2 $, consider the population version of $ \tilde \tau_j $:
$$ \tau_j^o = \frac{2}{K(K-1)} \sum_{1 \le g_1 < g_2 \le K}
\mbox{sign}(\mu_j^{g_2} - \mu_j^{g_1}). $$
Indeed, $ \mathbb E \tilde \tau_j \to \tau_j^o $ as $ N \to \infty $
(see Lemma S2 in the supplementary material).
For $ j \in J_{nom} $, suppose that
$ \mu^1_j < \dots < \mu^{g}_j < \mu^{g+2}_j$ and $\mu^{g-1}_j < \mu^{g+1}_j < \dots < \mu_j^K$ but
$ \mu^{g}_j > \mu^{g+1}_j $ for some $ 1 \le g \le K $.
In this case, $ |\tau_j^o| = 1 - \frac{4}{K(K-1)} $ and this achieves the maximum:
$ \max_{j \in J_{nom}}|\tau_j^o| = 1 - 4/\{K(K-1)\}.$
Since $ \mathbb E \tilde \tau_j  \to \tau_j^o $,
we have $ \tilde \tau_j < 1 - \theta_2 $ for large enough $ N $ if $ j \in J_{nom} $.
Conversely, if $ j \in J_{ord} $ then $ |\tau_j^o| = 1 $ which implies that
$ \tilde\tau_j \to |\tau_j^o| = 1 > 1 - \theta_2 $.
Thus, we can successfully separate ordinal and nominal variables by setting
$ \theta_2 $ as in Theorem \ref{thm: ordinal weight}.

In practice, we propose to use
\begin{equation} \label{eq: th-choices}
	\hat \theta_1 = \left(\frac{1}{2}\min_{j \in \hat J_{md}} |\hat \tau_j|\right)
	\vee \left(\max_{j \in \hat J_{noise}} |\hat \tau_j|\right), \quad
	\hat \theta_2 = \theta_2 = \frac{2}{K(K-1)},
\end{equation}
where $ \hat J_{md} $ and $ \hat J_{noise} $ are estimated index sets.
This choice of $\hat\theta_1$ is obtained by replacing $J_{md}$ and $ \Delta $ in the expression of $\theta_1$ 
with corresponding empirical quantities. 
To decide whether $ j \in \hat{J}_{md} $, a hypothesis test for testing $H_0^{(j)}: \mu_j^1 = \dots = \mu_j^K $ is conducted for each $ j \in J $. We adopt the usual $F$-test in our implementation. More precisely, we set $ j \in \hat J_{md} $ if the $F$-test rejects $ H_0^{(j)} $ at a pre-specified significance level, and we set $ j \in \hat J_{noise} $ if the $F$-test does not reject $ H_0^{(j)} $. In the numerical studies presented in Section \ref{sec: numerical study}, each variable is tested with a significance level of $ \alpha = 0.05$.

One may utilize a multiple testing correction for the simultaneous application of $F$-tests, thus controlling the false discovery rate (FDR)  or the familywise error rate (FWER), particularly in scenarios with a large $p$. However, controlling the FWER or FDR is not deemed essential for our ultimate purpose of learning ordinal basis. Uncorrected $F$-tests often lead to an excess of noise variables in $\hat{J}_{{md}}$, the detected mean difference set. Nevertheless, the two-step ordinal weight construction is robust to such false positives. This is because the second step of identifying $J_{ord}$ tends to set the corresponding weight to zero, mitigating the impact of false positives. Moreover, even in cases where a noise variable is assigned a high weight $w_j = 1$, the noise variable is inclined to be discarded during the basis learning procedure \eqref{eq: SOBL}, thanks to the utilization of the sparsity-inducing penalty term.
 
The two-step ordinal weights with $\hat \theta_1$ and $ \hat \theta_2 $ of (\ref{eq: th-choices}) effectively separate the ordinal weights in our numerical studies, as demonstrated next.

\subsection{Comparison of ordinal weights} \label{sec. comparing weights}
We present a numerical comparison of the ordinal weights, based on Spearman's rank correlation, Kendall's $ \tau $, the $t$-test-based trend test, and the two-step procedure.
We consider $K=4$ groups where the first five are all ordinal variables, next five are nominal variables, and the rest of variables have no mean difference. 
We assume Gaussian model, $ X\:|\:Y = j \sim N(\mu_j, \Sigma)$, with $p=200$, and 
generate $ 50 $ random samples per group and calculate ordinal weights. The ordinal weights, averaged over 50 repetitions, are shown in
Figure \ref{fig: ordinal weights}. (See Section S3 of the supplementary material for simulation setting.)

The two-step ordinal weight $w_j$ is much more superior than any other ordinal weights, in the sense that the average weight difference between ordinal variables and the others is the largest.
The  ordinal weights $w_j^P$ based on the $t$-test-based trend test (denoted as TSTT in the figure) perform almost similar to the two-step ordinal weights for mean difference variables.
However, for noise variables, $w_j^P$ is higher than the other three ordinal weights.
The weights based on rank correlations are similar to each other, and, as expected, the gap between weights of ordinal variables and weights of non-ordinal variables is thin. 

\begin{figure}[t]
	\begin{center}
		\includegraphics[width=0.7\textwidth]{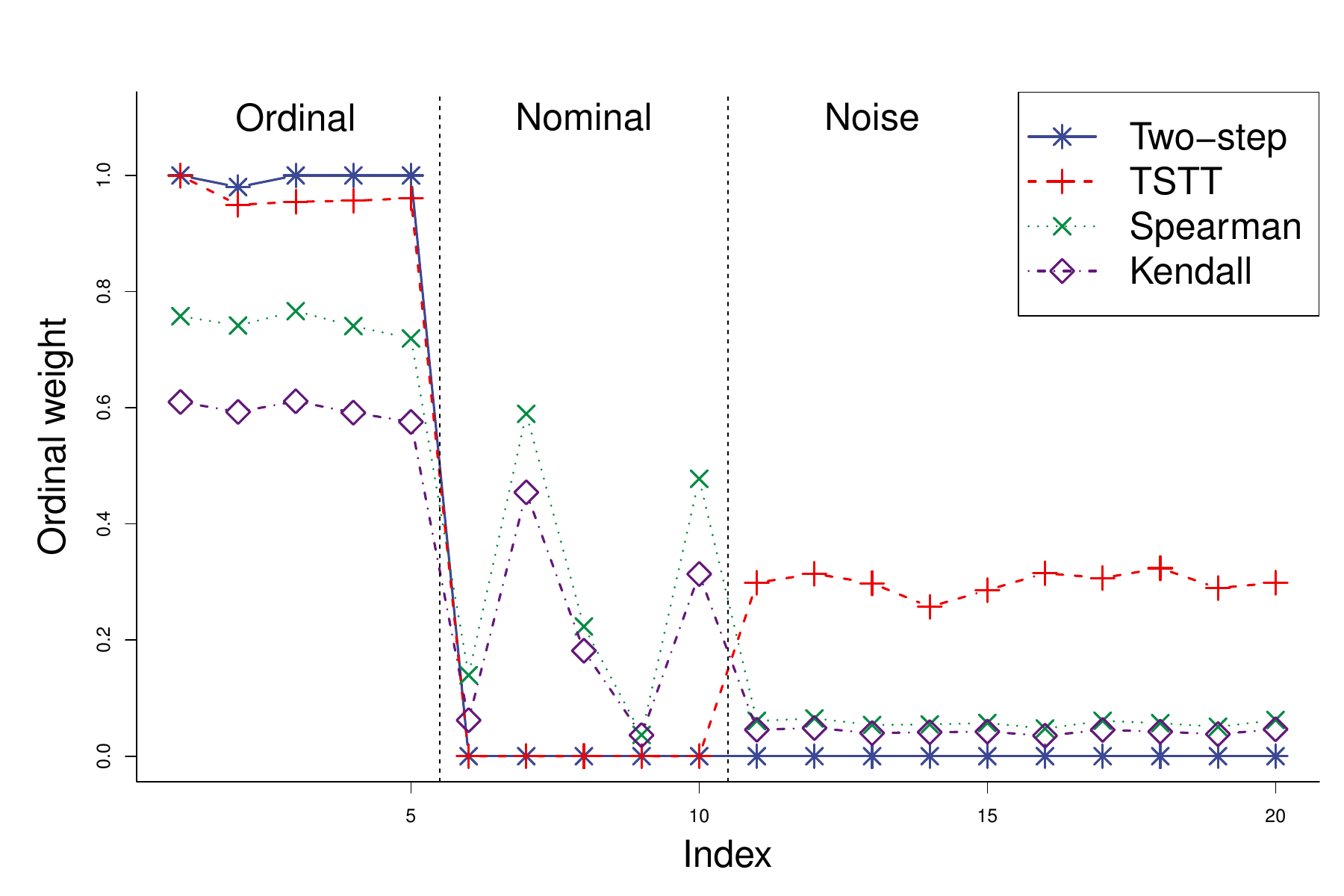}
		\caption{Averaged ordinal weights for the first $ 20 $ variables (The patterns for other $p-20$ variables are similar to the noise variables shown here).
		See Section \ref{sec. comparing weights}.}
		\label{fig: ordinal weights}
	\end{center}
\end{figure}

\section{Variable selection and basis learning properties}  \label{sec: theory}
We begin with introducing notation for matrix indexing and matrix norms.
For a matrix $V \in \mathbb{R}^{m \times n}$, $V_{i, :}$ denotes the $i$th row of $V$ and
$V_{:, j}$ denotes the $j$th column of $V$.
For index sets $I \subset \{1, 2, \dots, m\}$ and $J\subset \{1, 2, \dots, n\}$,
$V_{I, :}$ is a submatrix of $V$ with rows indexed by $I$.
We also denote $ V_{I} = V_{I, :} $ if there is no confusion.
$V_{I, J}$ denotes the submatrix of $V$ with row indices $I$ and column indices $J$.
Matrix norms $ \|\cdot\|_{\infty} $ and $ \|\cdot\|_{\infty, 2} $ are defined as
$\Vert V \Vert_{\infty} = \max_{1\le i \le m} \|V_{i,:}\|_{1} $ and
$\Vert V \Vert_{\infty, 2} = \max_{1\le i \le m} \|V_{i,:}\|_{2} $.
Next, for two sequences $ \{a_n\} $ and $ \{b_n\} $, we write $ a_n \lesssim b_n $ if there exists a positive constant $ C > 0  $ and $n_0 \in \mathbb N$ such that
$ a_n \le C b_n $ for all $ n \ge n_0 $. 
We denote $ a_n \asymp b_n $ if $ a_n \lesssim b_n $ and $ b_n \lesssim a_n $.

In this section, we give theoretical results for SOBL. Theorems presented in this section provide us with an understanding of the role of parameters $(\lambda, \eta)$ in SOBL. Under suitable choices of $(\lambda, \eta)$, statistical properties of SOBL, such as variable selection property and basis error bounds in terms of $\|\cdot\|_{\infty, 2}$ are presented in both non-asymptotic and asymptotic cases. Also, we investigate the variable selection property of OSBL basis. Proofs are given in supplementary material S1.2. 

%
\subsection{Sparse ordinal basis learning} \label{subsec: theory}

The proposed SOBL (\ref{eq: SOBL}) takes $ (\widehat \Sigma, \widehat M) $ as an input. 
Among the candidates of $ (\widehat \Sigma, \widehat M) $, we analyze the choices corresponding to MGSDA \citep{gaynanova2016sparseLDA} and MSDA \citep{mai2019multiclass}. 
While all results in this section apply to both choices, we do not present theoretical results corresponding to fastPOI \citep{jung2019poi}. This is because our strategy of dealing with mean differences does not directly apply to the eigenvectors of $ \widehat \Sigma_b $.

For the ordinal weights, we use the two-step ordinal weights, discussed in Section~\ref{subsec: two-step}, with
$ (\theta_1, \theta_2) $ defined in Theorem \ref{thm: ordinal weight}.
Although $ \theta_1 $ contains population quantities which we do not know in general,
Theorem \ref{thm: ordinal weight} enables us to derive the theoretical results presented in the current section.

For notational convenience, we let $ A = J_{disc} $, $ A_1 = J_{disc}^{ord} $, and $ A_2 = J_{disc} \cap J_{ord}^c $. Note that in general $A_2 \neq J_{nom} = J_{md} \cap J_{ord}^c$. Under these notation, we introduce population quantities that appear in the theorems. Let $ \kappa = \|\Sigma_{A^c A} (\Sigma_{A A})^{-1}\|_{\infty} $ stand for the irrepresentability quantity. Define $ \phi = \|(\Sigma_{AA})^{-1}\|_{\infty} $, $ \delta = \|M\|_{\infty, 2} $, $ \delta_1 = \|M_{A_1}\|_{\infty, 2} $ and $ \delta_2 = \|M_{A_2}\|_{\infty, 2} $. Let $ d = |J_{disc}| $.

The following theorem shows the variable selection property of SOBL.
\begin{thm} \label{thm: selection property}
	Assume conditions (C1)-(C4) in Section \ref{subsec: ts-theory} and $ \kappa < 1 $. Suppose that $\lambda > 0$ and $\eta \ge 1$.
	Let $ \widehat{D}^{sobl} = \widehat{D}(\widehat Z^{ord}_{\eta, \lambda}) $ be the selected variables 
	with ordinal weights defined in Theorem \ref{thm: ordinal weight}. 
        Then, the following hold for all sufficiently large $N \gtrsim (\vartheta_{\min}^2 \wedge 1)^{-1}$,
	\begin{enumerate}
		\item
		If $\lambda < \min(\frac{1}{2\phi}, \frac{(1+\kappa)(\phi^{-1} + \delta)}{2})$ and $1 \le \eta \le (\kappa^{-1}+1)/2 $, then
		$  
		\Pr\left(\widehat{D}^{sobl} \subset J_{disc} \right) 
		\ge
		1 - C\gamma,
		$
		where 
		$  
		\gamma 
		\equiv \gamma(N, p, d, \lambda)
		= pd e^{-C\min (\frac{\lambda^2}{d^2\sigma^4 K^2}, \frac{\lambda}{d \sigma^2 K})N}
		+ p e^{-\frac{CN\lambda^2}{d^2\sigma^2}} 
		+ p e^{-\frac{N \Delta^2}{32}}
		+ p e^{-CN(\vartheta_{\min} \land \Delta)^2}.
		$

		\item
		Suppose that $ \Sigma_{A_2 A_2^c} = \mathbf  0 $.
		If $ \lambda < \min(\frac{1}{2\phi}, \frac{2\beta(1+\phi\delta_1)(1+\kappa)}{(2+\phi(\beta' + \delta_1))(1-\kappa)}) $ and $ \beta \le \lambda\eta - \delta_2 \le \beta' $ for some constants $ \beta, \beta' > 0 $, then
		$ 
		\Pr\left(\widehat{D}^{sobl} \subset J_{disc}^{ord} \right)
		\ge 1 - C\gamma.
		$
	\end{enumerate}
\end{thm}

For the choice of $ \lambda $ that ensures consistent variable selection, see the asymptotic analysis in 
Section \ref{subsec: asymptotic}.

The first part of Theorem \ref{thm: selection property} states that when $ \eta \approx 1 $,
SOBL successfully screens out all variables in $ J_{disc}^c $
with probability at least $ 1-C\gamma $.
For $\eta = 1$, SOBL is equivalent to 
non-ordinal sparse LDA basis learning.
For such a special case, the finding of Theorem \ref{thm: selection property} is comparable to
the variable selection properties appeared in
\cite{gaynanova2016sparseLDA} and \cite{mai2019multiclass}.
This makes sense since for a small $ \eta $, there is no significant difference
between ordinal and non-ordinal sparse LDA methods.

The second part of Theorem \ref{thm: selection property} shows that
with a suitably large choice of
$\eta$, SOBL only chooses the variables contained in $ J_{disc}^{ord}$
and does not select variables in $ (J_{disc}^{ord})^c $ with high probability.
If a pair $ (\lambda, \eta) $ satisfies the conditions of the second part of the theorem, then
it also satisfies 
$ \lambda\eta > \delta_2 = \|M_{A_2}\|_{\infty, 2}$.
Under the setting of of the theorem, for $ j \in J_{ord}^c $, the penalty coefficient $ \lambda_j = \lambda \eta^{1-w_j} $ becomes $ \lambda \eta $ with high probability. Thus, for $ \widehat Z^{ord}_{\eta, \lambda} $ to choose only $ J_{disc}^{ord} $,
one should choose $ (\lambda, \eta) $ large so that $ \lambda\eta $ is larger than
the mean differences of variables in $ A_2 = J_{disc} \cap J_{ord}^c $.

Note that the seemingly restrictive assumption $ \Sigma_{A_2 A_2^c} = \mathbf 0 $ in the second part of the theorem is only for simplicity. In fact, for the probability to converge to $1$, it is enough to set $ \Sigma_{A_2 A_2^c} $ sufficiently small and converging to $ \mathbf 0 $ with a suitable order of convergence as $ N,p \to \infty $. To properly select the ordinal discriminant variables, it is necessary that the variables of $ A_2 = J_{disc} \cap J_{ord}^c $ are (nearly) uncorrelated with the other variables.

Based on Theorem \ref{thm: selection property}, the next theorem states a non-asymptotic bound for the
error of the estimated SOBL basis $ \widehat Z_{\eta, \lambda}^{ord} $ in estimation of the population discriminant basis $ \Psi $.

\begin{thm} \label{thm: bound}
	Assume conditions (C1)-(C4) in Section \ref{subsec: ts-theory} and $ \kappa < 1 $.
	Let $ \widehat Z^{ord}_{\eta, \lambda} $ be the SOBL basis with ordinal weights defined in Theorem \ref{thm: ordinal weight}.
        The following hold for all sufficiently large $N \gtrsim (\vartheta_{\min}^2 \wedge 1)^{-1}$,
	\begin{enumerate}
		\item
		Suppose that $ \lambda $ and $ \eta $ satisfy same condition in the first part of Theorem 2. Then
		$  
		\Pr	(\| \widehat Z^{ord}_{\eta, \lambda} - \Psi\|_{\infty, 2} < 6\phi\eta\lambda
		) \ge 1 - C\gamma
		$
		
		\item 
		Suppose that $ \lambda $ and $ \eta $ satisfy same condition in the second part of Theorem 2. Then
		$  
		\Pr(\| (\widehat Z^{ord}_{\eta, \lambda})_{J_{disc}^{ord}} - \Psi_{J_{disc}^{ord}}\|_{\infty, 2}
		\le 6\phi\lambda) \ge 1 - C\gamma.
		$
		In addition, we have that
		$\| (\widehat Z^{ord}_{\eta, \lambda})_{J_{disc}^c} - \Psi_{J_{disc}^c}\|_{\infty, 2}
		= 0$, and
		$
		\| (\widehat Z^{ord}_{\eta, \lambda})_{A_2} - \Psi_{A_2}\|_{\infty, 2}
		\le \phi\delta_2
		$
		hold with probability at least $ 1 - C\gamma $.
	\end{enumerate}
\end{thm}

Similar to the case of variable selection, for a small $ \eta \approx 1 $,
$ \|\widehat Z_{\eta, \lambda}^{ord} - \Psi\|_{\infty, 2} $ is less than $\lambda $ up to constant with
probability greater than or equal to $ 1-C\gamma $.
Similar results were observed in sparse multiclass LDA contexts \citep{gaynanova2016sparseLDA,mai2019multiclass}.

The second part of Theorem \ref{thm: bound} reveals that the SOBL basis restricted to
$ J_{disc}^{ord} $ estimates $ \Psi_{J_{disc}^{ord}} $ with the error bounded by a multiple of  $ \lambda $. 
In this case, however, SOBL discards the variables in $ A_2 = J_{disc} \cap J_{ord}^c $,
which also affects the population discriminant subspace.
In this perspective, classification performance may decline when we use SOBL rather than a sparse LDA.
On the other hand, if inherent ordinal signals are strong enough,
then the classification performance of SOBL is comparable to that of a sparse LDA, while SOBL
selects significantly fewer variables than sparse LDA. In  Section \ref{sec: numerical study}, we will numerically investigate the situations under which the sparse and ordinal basis, estimated by SOBL, provides better interpretability  while showing comparable classification performances.

\subsubsection{Asymptotic analysis} \label{subsec: asymptotic}

For asymptotic analysis, we assume that $p, d$ may grow as $ N\to \infty $. We also assume that the tuning parameters $ \lambda \equiv \lambda_N $ and $ \eta \equiv \eta_N $ depend on the sample size $N$.
To guarantee consistent variable selection of $ \widehat Z_{\eta_N, \lambda_N}^{ord} $ (Theorem \ref{thm: selection property}), $ \gamma(N, p, d, \lambda_N)$ should converge to $ 0 $ as $ N \to \infty $.  
For the consistent basis learning (Theorem \ref{thm: bound}), we need to set $ \lambda_N \to 0 $ as $ N \to \infty $. 
However, if $ \lambda_N$ decreases excessively fast, then $ \gamma $ may not be converge to $ 0 $.
Thus, $ \lambda_N $ must decrease to $ 0 $ in a moderate speed with respect to the asymptotic conditions of $ (N, p, d) $.
In a similar manner, we need asymptotic conditions on $ \Delta $ and $  \vartheta_{\min} $ for the second and third terms of $ \gamma $.

To develop asymptotic versions of Theorems \ref{thm: selection property} and \ref{thm: bound}, the following asymptotic conditions are required.
\begin{itemize}
	\item[(AC1)] $ \frac{\log(pd)d^2}{N} \to 0$.
	\item[(AC2)] $ \lambda_N \gtrsim \left( \frac{\log(pd)d^2}{N} \right)^{\frac{1-\alpha}{2}} $
	for some $ \alpha \in (0, 1) $ and $ \lambda_N \to 0 $.
	\item[(AC3)] $ \Delta \wedge \vartheta_{\min} \gtrsim \frac{1}{d} $.
	\item[(AC4)] $\beta < \lim_{N \to \infty} \eta_N\lambda_N - \delta_2 < \beta' $ for some constants $ \beta, \beta' > 0 $.
\end{itemize}

Note that (AC1) and (AC2) also appear in the asymptotic results of sparse LDA
\citep{mai2012direct,gaynanova2016sparseLDA, mai2019multiclass}.
Condition (AC1) implies that $ p $ may grow faster than any polynomial order of
$N$ \citep{mai2012direct}.
(AC2) controls the error bound for the estimated basis with probability tending to $1$.
Condition (AC3) is needed to guarantee consistency of the two-step ordinal weights.
We note that $\Delta \wedge \vartheta_{\min} \to 0$ is allowed in (AC3) if $d \to \infty$.
Finally, (AC4) is only required when the variable selection is targeted at
$ J_{disc}^{ord} $, as in the second part of Theorem \ref{thm: selection property}.

With conditions (AC1)-(AC4), SOBL enjoys consistency in variable selection and
of the estimated basis, as shown in the following two corollaries. 
\begin{cor} \label{cor: AS1}
	Suppose conditions (C1)-(C2) and asymptotic conditions (AC1)-(AC3) hold.
	Assume that $ \kappa < 1 $ and $ 1 < \eta_N < \frac{\kappa^{-1}+1}{2} $ for any $ N $.
	Then, as $N\to \infty$, 
	$
	\Pr\left(\widehat{D}^{sobl} \subset J_{disc} \right)  \to 1
	$
	and
	$
	\Pr\left(\| \widehat Z^{ord}_{\eta, \lambda} - \Psi\|_{\infty, 2} <
	\varepsilon
	\right) \to 1
	$
	for any $ \varepsilon > 0 $.
\end{cor}

\begin{cor} \label{cor: AS2}
	Suppose conditions (C1)-(C2) and asymptotic conditions (AC1)-(AC4) hold.
	Assume that $ \kappa < 1 $ and $ \Sigma_{A_2 A_2^c} = \mathbf 0 $ for any $ p $.
	Then, as $N\to \infty$,  
	$
	\Pr\left(\widehat{D}^{sobl} \subset J_{disc}^{ord} \right)  \to 1
	$
	and
	$
	\Pr\left(\| (\widehat Z^{ord}_{\eta, \lambda})_{J_{disc}^{ord}} - \Psi_{J_{disc}^{ord}}\|_{\infty, 2} <
	\varepsilon
	\right) \to 1
	$
	for any $ \varepsilon > 0 $.
\end{cor}

\subsection{Ordinality-screened basis learning}
In this section, we briefly investigate consistent variable selection properties of the OSBL basis $\widehat Z_{\lambda}^{scr}$ introduced in Section \ref{sec:2.3}.
We set the threshold $ \theta_w = 1 $, and use the two-step ordinal weights as used in the previous section.
Based on Theorems \ref{thm: ordinal weight} and \ref{thm: selection property}, we obtain a similar variable selection consistency results for $\widehat Z_{\lambda}^{scr}$.
Moreover, if the sub-Gaussian model \eqref{eq:subG-model} is replaced by a more strict Gaussian model, then the OSBL basis possesses 
better asymptotic rates of consistency than those of SOBL, when the number $d$ of discriminant variables is increasing. 
The additional result is 
obtained by 
the consistent variable selection result of the MGSDA (Theorem 2 of \cite{gaynanova2015optimal}) and Theorem \ref{thm: ordinal weight}.

\begin{cor} \label{cor: osbl-selection}
	Let $ \widehat{D}^{osbl} = \widehat{D}(\widehat Z^{scr}_{\lambda}) $ be the set of variables selected by OSBL basis with ordinal weights defined in Theorem \ref{thm: ordinal weight}.
	\begin{enumerate}
		\item  Suppose conditions (C1)-(C2) and asymptotic conditions (AC1)-(AC3) hold. Assume that $ \kappa < 1 $. Then, as $N\to \infty$, $ \Pr\left(\widehat{D}^{osbl} \subset J_{disc}^{ord} \right) \to 1 $.

		\item Assume that $ X \mid Y=g \sim N(\mu_g, \Sigma_w) $ for each $g$. Consider a MGSDA based OSBL basis, and suppose conditions (C1)-(C2) and asymptotic condition (AC3). Further, suppose that $\min_{j \in J_{disc}}\|\Psi_{j}\|_2 \gtrsim \sqrt{\frac{\log(pd)}{N}} $, $ \lambda_N \gtrsim \sqrt{\frac{\log(pd)}{N}} $ and $ \frac{\log(pd)d}{N} \to 0$. Then, as $N\to \infty$,  
        $ 
        \Pr\left(\widehat{D}^{osbl} = J_{disc}^{ord} \right) \to 1.
        $
	\end{enumerate}
\end{cor}

The first part of the corollary says that the consistent variable selection property also holds for 
$\widehat{D}^{osbl}$
under the same conditions for SOBL. Furthermore, the OSBL basis exactly recovers ordinal discriminant variables with high probability. 

On the other hand, if the more strict Gaussian assumption is imposed, then consistency holds under the asymptotic regime $ \log(pd)d/N \to 0 $, which is better than (AC1) by a factor of $ d $. 
To obtain this better sample complexity, it is critical to utilize 
the distributional properties of $\widehat{\Sigma}^{-1}\widehat{M}$ and $\|\widehat{M}^{\top}\widehat{\Sigma}^{-1}\widehat{M}\|_2$, characterized by \cite{gaynanova2015optimal}, that are obtained under under the Gaussian model assumption.
In comparison, the theoretical properties of SOBL in Theorem \ref{thm: selection property} are based on elementwise concentration results for $\widehat{\Sigma}$ and $\widehat M$, which are needed to handle the ordinal penalty term $\lambda\eta^{1-w_j}$ in \eqref{eq: SOBL}. 
It has not been straightforward to extend the technical arguments of
\cite{gaynanova2015optimal} to our situation. 
In the next section, we empirically compare the two proposed methods with simulation and real data examples. Although the OSBL basis has theoretical benefits, our numerical results suggest that SOBL is preferable.

\section{Numerical studies} \label{sec: numerical study}

In this section, we numerically demonstrate the performances of variable selection and classification of the proposed methods. Once a sparse ordinal basis is obtained from either SOBL or OSBL, one can conduct classification or visualization using the data projected onto the low-dimensional subspace.
Specifically, we perform QR decomposition of $\widehat Z_{\eta, \lambda}^{ord}$ or $\widehat Z_{ \lambda}^{scr}$, as done in \cite{jung2019poi}, to get an orthogonal basis $\hat Q$ that spans the low-dimensional subspace $\mathcal C(\widehat Z_{\eta, \lambda}^{ord})$ or  $\mathcal C(\widehat Z_{ \lambda}^{scr})$. The classification rule is then fitted by applying classical multiclass LDA on $(\mathbf{X} \hat Q, \mathbf{Y})$. Also, the projected dataset $(\mathbf{X} \hat Q, \mathbf{Y})$ can be used to visualize the low-dimensional projection via a scatter plot \citep{gaynanova2016sparseLDA}.

In the rest of the section, for SOBL and OSBL, we only consider the two-step ordinal weights with the thresholds \eqref{eq: th-choices}. In the implementation step, we add a small $\epsilon I_p$ to $\widehat \Sigma$ for numerical stability when we compute SOBL or OSBL basis.

\subsection{Tuning parameter selection scheme} \label{subsec: tuning}
For the parameter tuning of the OSBL basis $ \widehat{Z}_{\lambda}^{scr} $, we only need to decide the value of $ \lambda $. Following the common practice in the sparse multiclass LDA literature, we conducted a grid search to identify the optimal value of $ \lambda $ that maximizes classification accuracy among the candidate values. 

Similarly, for the SOBL basis $ \widehat{Z}_{\eta, \lambda}^{ord} $,
a standard approach for the tuning parameters involves cross-validation on
a predetermined $ (\lambda, \eta) $ grid. 
However, as we will see in our simulation results, while such a choice enhances classification accuracy, it
typically results in selecting non-ordinal variables as well as ordinal discriminant variables. 
For better interpretability, we aim to capture the ordinal discriminant variables in $ J_{disc}^{ord} $.
For this purpose, we propose a two-step tuning procedure for determining $(\tilde \lambda, \tilde \eta)$:
\begin{enumerate}
	\item Fix $ \eta = 1 $. Choose the best $ \tilde\lambda $ on a fine grid
	$ (0, \lambda_{\max}) $ with respect to the classification accuracy.
	\item 
    Fix $ \lambda = \tilde \lambda $ and fit $ \widehat Z_{\eta, \tilde\lambda}^{ord} $ for a sequence of $\eta$, ranging from $ 1 = \eta_0 < \eta_1 < \dots < \eta_m =  \eta_{\max}$. 
    Let $ \tilde m = \min \{i: \|\widehat Z_{\eta_i, \tilde\lambda}^{ord} - \widehat Z_{\eta_{i-1}, \tilde\lambda}^{ord}\|_F < 10^{-10}\}$ and choose $\tilde \eta = \eta_{\tilde m}$.
\end{enumerate}
We set $ \lambda_{\max} = \|\widehat M\|_{\infty, 2} $,
and $ \eta_{\max} = 2(\|\widehat M\|_{\infty, 2}/\tilde{\lambda} + 1) $.
If $ \lambda > \lambda_{\max} $ then $ \widehat Z_{1, \lambda}^{ord} $ becomes
$ 0 $ \citep{jung2019poi}.
The choice of $ \eta_{\max} $ comes from the assumption for $ \eta $ in
the second part of Theorem \ref{thm: selection property}, that is, $ \lambda\eta - \delta_2  \asymp 1 $ and $ \eta > 1 $.
The basic idea of the two-step tuning procedure is based on the second part of
Theorem \ref{thm: selection property} and the asymptotic condition (AC2):
To choose $ J_{disc}^{ord} $ properly, we need large enough $ \eta $ since proper $ \tilde\lambda $ tends to $ 0 $ as $ N $ increases. 
Unless otherwise noted, the numerical results of SOBL in this section assumes the two-step tuning procedure. 

How the two-step tuning procedure works is illustrated with a toy data in  Figure \ref{fig: cv-process}.
The data are generated from Model III of Section \ref{subsec: sim-var}, in which we set $ p = 200 $ and $ (n_1, n_2, n_3) = (100, 100, 100)$.
The SOBL, based on fastPOI, is fitted for this data set.
In the first step of the tuning procedure, $ \tilde \lambda = 0.277$ (with $\eta =1$ fixed) provides the largest cross-validated classification accuracy. With $ \tilde \lambda$ fixed, we observe that as $\eta $ increasing,  non-ordinal variables tend to ``die'' while
the ordinal discriminant variables in $ J_{disc}^{ord} $ resist and stay included in the estimate for larger $\eta$. 
The left panel of Figure \ref{fig: cv-process} shows that as $ \eta $ increases,
$ \|(\widehat{Z}_{\eta,\tilde\lambda}^{ord})_{j,:}\|_{2}$ decreases to $ 0 $ for $j \in J_{ord}^c$.
Remarkably, one of the two variables belonging to $ J_{disc}^{ord} $  was not selected for small
values of $ \eta $, then became included for larger $ \eta > 2 $.
After all nominal and noise variables are removed at $\eta = \tilde\eta = 2.26$,
there is no further change in the estimate as $ \eta $ continues to increase.
This is because the penalty coefficient for ordinal variables,
$ \lambda_j = \tilde \lambda $, is fixed throughout the second tuning step,
while the penalty coefficient for non-ordinal variables,
$ \lambda_j =  \tilde \lambda \eta $, grows as $ \eta $ increases. Thus, once the non-ordinal variables become zero, a larger penalty simply does not make a difference in estimates.
This two-step tuning procedure aims to select $ J_{disc}^{ord} $, albeit
making compromises on classification accuracy, if the ordinal signal strength is weak.
As can be seen in the scatter plots in Figure \ref{fig: cv-process}, the  data projected on the ordinal basis $\widehat{Z}_{\tilde\eta,\tilde\lambda}^{ord}$ (shown in the right panel) clearly show an ordinal group tendency in the first direction, while the data projected on $\widehat{Z}_{1,\tilde\lambda}^{ord} = \widehat{Z}_{\tilde\lambda}$ in the middle panel do not reflect the order.

\begin{figure}[t]
	\begin{center}
		\includegraphics[width=\textwidth]{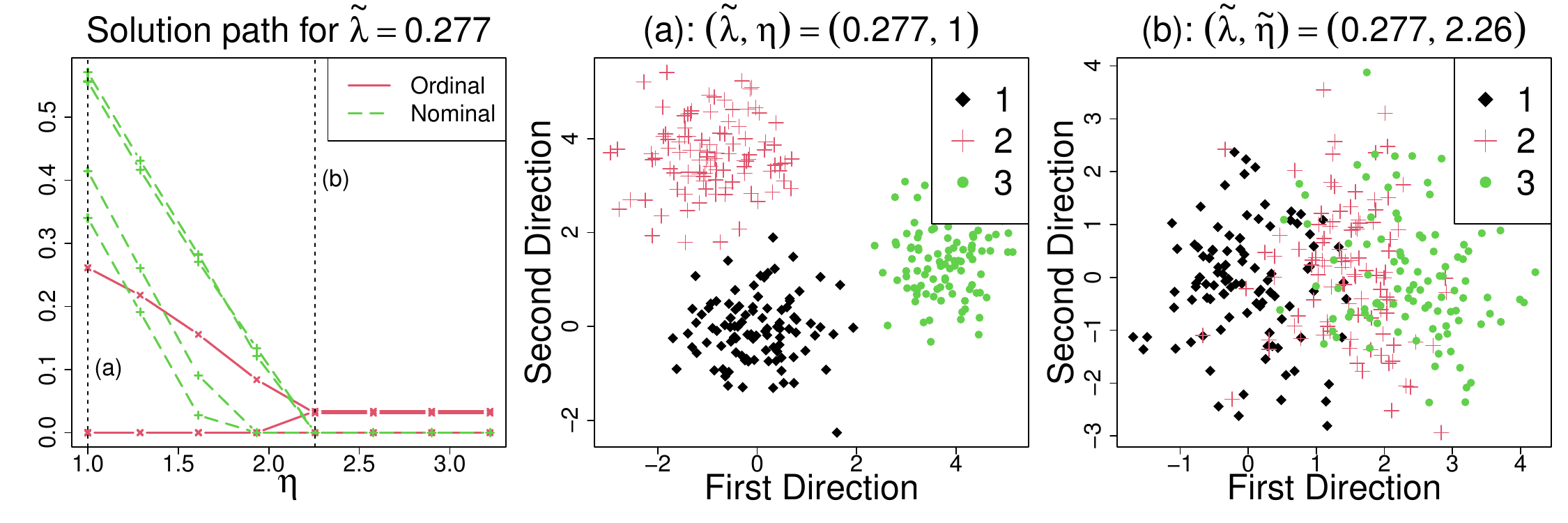}
		\caption{%
		The left panel shows a solution path for
		$ \tilde\lambda=0.277 $ of SOBL based on fastPOI with respect to $ \eta $.		
		Each line denotes $ \|(\widehat Z_{\eta, \tilde\lambda}^{ord})_{j, :}\|_2 $ for
		$ j \in J_{md} $.
%
%
%
%
%
%
		The middle and right panels show corresponding projected scatter plots for (a) and (b)
		in the first plot.
		}
		\label{fig: cv-process}
	\end{center}
\end{figure}

\subsection{Competing methods} \label{subsec: competing}
We compare variable selection and classification performances of our proposed ordinal basis learning to the methods of
\cite{ma2021bioinfo} and \cite{archer2014ordinalgmifs}.
\cite{ma2021bioinfo} proposed feature-weighted ordinal classification (FWOC)
which solves the sparse LDA objective (\ref{eq: sparse LDA}) with $\widehat \Sigma$ adjusted by ordinal weights. The FWOC is also a sparse basis learning method.
In \cite{archer2014ordinalgmifs}, several types of penalized ordinal logistic regression
models are proposed.
Among those, we choose to compare with a penalized cumulative logistic model (PCLM), which estimates the coefficient vector $\beta \in \mathbb{R}^p$ maximizing a penalized likelihood $ L(\alpha, \beta) -\lambda \sum_{j=1}^p |\beta_j| $ under
the cumulative logit model:
$\mbox{logit}(\mbox{Pr}\left( Y \le j \:|\: X \right))=\alpha_j+\beta^\top X $.
We follow tuning procedures outlined in the original papers.

\subsection{Simulation study} \label{subsec: sim-var}
The variable selection and classification properties of the proposed methods are demonstrated with a simulation study. 
For the numerical simulation, consider the mean vectors of three groups defined as  
\begin{align*}
    \mu_1^\top 
    &= (\overbrace{s_{ord} \cdot (1, 0, 0,\:\:\:~0)}^{J_{ord}},
    ~\overbrace{s_{nom} \cdot (0,\:\:\:\:\:\:~~0, \:\:\:~0, \:\:\:\:\:\:\:0)}^{J_{nom}},
    ~\overbrace{\mathbf{0}_{p-8}^\top}^{J_{noise}}), \\
    \mu_2^\top 
    &= (s_{ord} \cdot (2, 1, 2, -2),
    ~s_{nom} \cdot (3,\:\:\:\:\:\:~~2, -1, -0.5),
    ~\mathbf{0}_{p-8}^\top), \\
    \mu_3^\top 
    &= (s_{ord} \cdot (3,2, 4, -3),
    ~s_{nom} \cdot (2, -0.5, \:\:\:2, \:\:\:\:\:\:\:3),
    ~\mathbf{0}_{p-8}^\top).
\end{align*}
Here, $ (s_{ord}, s_{nom}) $ are scale constants for ordinal and nominal variables, which vary in different model settings. We set common within covariance matrix as $\Sigma_w = \mbox{diag}(\Sigma_8, I_{p-8})$, where $\Sigma_j = \frac{1}{2}(I_j + 1_j1_j^\top)\ \mbox{for} \ j \in  \mathbb N.$
In the simulation, we used a fixed design data generation process.
For the $g$th group, we generate random samples from $ N_p(\mu_g, \Sigma_w) $ with $ p=800 $.
We repeat 100 times to generate  $(n_1, n_2, n_3) =  (50, 50, 50) $ random samples
for the train set and generate a test set with the same sample sizes.
For a choice of scale constants, we consider three models as follows. 
Model I: $ (s_{ord}, s_{nom}) = (1, 0) $, $ |J_{disc}|=8 $ and $ |J_{disc}^{ord}|=4 $.
Model II: $ (s_{ord}, s_{nom}) = (0.75, 0.5) $, $ |J_{disc}|=8 $ and $ |J_{disc}^{ord}|=4 $. 
Model III: $ (s_{ord}, s_{nom}) = (0.5, 1) $, $ |J_{disc}|=6 $ and $ |J_{disc}^{ord}|=2 $. 
As models vary from I to III, the magnitudes of mean differences for ordinal variables become small while those of nominal variables become large.
Note that Model III is the same as the mean structure that appeared in Example \ref{example-1}.
Tuning parameters for our proposed methods are selected as described in Section \ref{subsec: tuning}.
In addition, for SOBL, we also conducted a  5-fold cross-validation procedure on a $(\lambda, \eta)$ grid to compare with the two-step tuning procedure.
In Tables \ref{tab: varsel} and \ref{tab: sim-acc}, ``SOBL-grid'' stands for the results of SOBL, based on the two-dimensional $ (\lambda, \eta) $ grid search.
For a baseline comparison, we also fit and report the results from the sparse LDA method (either from \cite{gaynanova2016sparseLDA} or \cite{mai2019multiclass} or \cite{jung2019poi}), denoted as SLDA.
After fitting, we measure variable selection performance on the fitted model and
report classification performances with respect to $ l_0, l_1 $ and $ l_2 $ losses. Here, the $ l_k $ loss, for $ k = 0, 1, 2 $, is given by $\frac{1}{N} \sum_{i=1}^N |\hat Y_i - Y_i|^k$ where $\hat Y_i$ is an estimated class of the $ i $th observation in the test set. 
Note that $l_0$ loss is same as the misclassification rate.
In the ordinal response setting, using the $ l_1 $ or $ l_2 $ loss may be preferable to employing the $ l_0 $ loss, which only measures classification accuracy.
Even with the same values for the $ l_0 $ loss, the fitted result may possess a  lower $ l_1 $ or $ l_2 $ loss if the fitted $\{\hat Y_i\}_{i=1}^N $ exhibit some order-concordant tendency to $ \{Y_i\}_{i=1}^N $.

\begin{table}[t]
	\centering
	\resizebox{\columnwidth}{!}{
\begin{tabular}{ccc|ccccccc}
    Model & True & Index set & SOBL & OSBL & SOBL-grid & SLDA & FWOC & PCLM & Weight  \\ \hline
    \hline
    \multirow{4}{*}{I}
    & & $|\widehat D|$  & 3.36(0.131) &  3.37(0.130) &  16.71(3.824) &  18.56(4.466) & 27.49(4.700) &  9.14(0.440) & 11.0(0.312)\\
    & $|J_{disc}| = 8$ & $|\widehat D \cap J_{disc}|$   & 3.08(0.088) & 3.09(0.085) & 3.73(0.099) & 3.36(0.111) & 4.27(0.092) & 3.90(0.067) & 4.03(0.017)\\
    & $|J_{disc}^{ord}| = 4$ & $|\widehat D \cap J_{disc}^{ord}|$ & 3.08(0.088) & 3.08(0.085) & 3.52(0.072) & 3.17(0.084) & 3.84(0.047) & 3.70(0.048) &  4.00(0.000)\\
    & & $|\widehat D \cap J_{disc}^{ord}|/|\widehat D|$  & 0.95(0.011) & 0.95(0.012) & 0.65(0.035) & 0.68(0.038) & 0.54(0.039) & 0.49(0.021) & 0.40(0.014) \\
    \hline
    \multirow{4}{*}{II}
    & & $|\widehat D|$  & 3.79(0.090) &  3.58(0.121) &  15.81(1.265) & 14.12(0.988) & 34.41(4.452) & 15.59(0.791) & 10.7(0.351)\\
    & $|J_{disc}| = 8$ & $|\widehat D \cap J_{disc}|$  & 3.52(0.063) & 3.13(0.073) & 7.13(0.084) & 7.03(0.086) & 7.47(0.076) & 5.11(0.091) & 4.16(0.037)\\
    & $|J_{disc}^{ord}| = 4$ & $|\widehat D \cap J_{disc}^{ord}|$ & 3.42(0.065) & 3.03(0.076) & 3.26(0.072) & 3.10(0.075) & 3.50(0.069) & 3.55(0.059) & 4.00(0.00)\\
    & & $|\widehat D \cap J_{disc}^{ord}|/|\widehat D|$ & 0.92(0.012) & 0.88(0.017) & 0.28(0.013) & 0.29(0.012) & 0.24(0.016) & 0.30(0.016) & 0.42(0.015)\\
    \hline
    \multirow{4}{*}{III}
    & & $|\widehat D|$ &   1.44(0.087) &  1.35(0.096) &  9.82(2.58) &  9.02(2.503) & 22.86(7.333) & 17.98(1.156)  & 11.4(0.326)\\
    & $|J_{disc}| = 6$ & $|\widehat D \cap J_{disc}|$ &   1.41(0.082) & 1.17(0.037) & 5.38(0.051) & 5.13(0.040) & 5.36(0.056) & 4.51(0.060)  & 2.02(0.014)\\
    & $|J_{disc}^{ord}| = 2$ & $|\widehat D \cap J_{disc}^{ord}|$& 1.34(0.074) & 1.17(0.037) & 1.40(0.051) & 1.13(0.040) & 1.36(0.056) & 2.00(0.000)& 2.00(0.00) \\
    & & $|\widehat D \cap J_{disc}^{ord}|/|\widehat D|$  & 0.97(0.015) & 0.95(0.015) & 0.22(0.008) & 0.19(0.006) & 0.18(0.010) & 0.16(0.008) & 0.193(0.007)
\end{tabular}
}
	\caption{Variable selection performances of each method with 100 repetitions on simulation settings. Averaged measures are presented. Values in parentheses are standard errors. 
	\label{tab: varsel}}
\end{table}

In Tables \ref{tab: varsel} and \ref{tab: sim-acc}, we  report the case of proposed ordinal basis learning methods only based on fastPOI, since the results of the proposed methods based on other choices of sparse LDAs have similar patterns. 

Table \ref{tab: varsel} shows variable selection performance results. Each method's index set of selected variables is denoted by $ \widehat D $. Thus, $|\widehat{D}|$ is the number of selected variables. 
The rightmost column of the table represents the variable selection performance of variables selected by two-step ordinal weights, which is defined as $\widehat D = \{j: w_j = 1\}$.
In all three models, the proposed SOBL and OSBL show the highest ratio of $ J_{disc}^{ord} $ in selected variables, $ |\widehat D \cap J_{disc}^{ord}|/|\widehat D| $, among all methods considered. Although no nominal variables exist and only ordinal variables exist in Model I, SOBL and OSBL methods capture $ J_{disc}^{ord} $ more likely than the other methods. 
In Model II and Model III, SOBL's selected ratio $ |\widehat D \cap J_{disc}^{ord}|/|\widehat D| $ is slightly higher than OSBL's. In comparison, SOBL-grid selects almost all $ J_{disc} $, but some of the noise variables are falsely included. PCLM always selects all variables in $ J_{disc}^{ord} $ but also mistakenly selects too many of $ J_{disc}^c $. FWOC can not properly select $ J_{disc}^{ord} $ and also frequently selects noise variables. 
Similarly, the baseline method, SLDA, tends to include noise variables in all three models, and the selected ratio $ |\widehat D \cap J_{disc}^{ord}|/|\widehat D| $ of SLDA decreases as models vary from I to III.
Lastly, the ordinal weight always selects all variables in $J^{ord}_{disc}$ but also selects variables not included in $J^{ord}_{disc}$. Ordinal weight aims to capture $J_{ord}$, so it also selects variables in $J_{ord} \setminus J_{disc}$. Under Model III, $J_{ord} \setminus J_{disc}$ is a nonempty set, as seen in Example \ref{example-1}. Furthermore, utilizing only ordinal weights does not guarantee noise variables not to be selected. Based on these observations, we conclude that using ordinal weight for the variable selection is insufficient for selecting  variables in $J^{ord}_{disc}$.

The classification performance results are summarised in Table \ref{tab: sim-acc} in terms of $ l_k $ losses. In Model I, all classification performances are at a similar level. This is because there is no nominal variable. Remarkably, SOBL and OSBL are the only methods that provide comparable classification performances to other competing methods while only using the smallest numbers of variables contained in $ J_{disc}^{ord} $.  In Models II and III, SOBL and OSBL have higher losses than competing methods. 
This is because the magnitude of the mean difference of ordinal variables is smaller than that of nominal. 
%
Nevertheless, it is observed that the proposed methods properly select $ J_{disc}^{ord} $, in any case. Even if the ordinal signal strength is weak as in Model III, the proposed SOBL and OSBL properly identify $J_{disc}^{ord}$ with high probabilities. Finally, we note that the classification performance of SOBL-grid is among the best across all situations.

\begin{table}[t]
	\centering
	\resizebox{\columnwidth}{!}{
		\begin{tabular}{cc|cccccc}
			Model & Metric & SOBL & OSBL & SOBL-grid & SLDA & FWOC & PCLM  \\ \hline
			\hline
			\multirow{3}{*}{I}
   		&$l_0$&0.054(0.022) & 0.052(0.021) & 0.057(0.023) & 0.057(0.023) & 0.054(0.022) & 0.052(0.019) \\
			&$l_1$&0.054(0.022) & 0.052(0.021) & 0.057(0.023) & 0.057(0.023) & 0.054(0.022) & 0.052(0.019) \\
			&$l_2$&0.054(0.022) & 0.052(0.021) & 0.057(0.023) & 0.057(0.023) & 0.054(0.022) & 0.052(0.019)\\
			\hline
			\multirow{3}{*}{II}
			&$ l_0 $&0.118(0.027) & 0.121(0.028) & 0.036(0.015) & 0.036(0.015) & 0.038(0.015) & 0.114(0.026)\\
			&$ l_1 $&0.118(0.027) & 0.121(0.028) & 0.037(0.016) & 0.036(0.015) & 0.039(0.016) & 0.114(0.026)\\
			&$ l_2 $&0.119(0.027) & 0.121(0.029) & 0.040(0.018) & 0.038(0.017) & 0.041(0.018) & 0.114(0.026)\\
			\hline
			\multirow{3}{*}{III}
			&$l_0$&0.352(0.012) & 0.355(0.005) & 0.009(0.001) & 0.010(0.001) & 0.011(0.001) & 0.156(0.003)\\
			&$l_1$&0.391(0.017) & 0.376(0.006) & 0.013(0.001) & 0.015(0.001) & 0.017(0.002) & 0.156(0.003)\\
			&$l_2$&0.469(0.027) & 0.418(0.009) & 0.021(0.002) & 0.025(0.003) & 0.030(0.003) & 0.156(0.003)
	\end{tabular}
	}
	\caption{Classification performances of each method with
		100 repetitions on simulation settings. Averaged measures are presented.
		Values in parentheses are standard errors.
		\label{tab: sim-acc}}
\end{table}

\subsection{Real data analysis} \label{subsec: real-data}
We apply the proposed methods to two high-dimensional gene expression data sets with ordinal responses. 
The first data set is the primary human  Glioma data set \citep{sun2006neuronal}.
Glioma is a type of tumor that starts in the brain and spine and has severe prognosis results
for patients.
In the Glioma data set, each observation has an ordinal class among
`Normal' $ \prec $ `Grade II' $ \prec $ `Grade III' $ \prec $ `Grade IV.'
The predictor $ X $ is preprocessed gene expression data with $ p = 54,612 $.
The sample sizes of the four classes are $ (n_1, n_2, n_3, n_4) = (23, 45, 31, 81) $,
with $ N =180 $.
The second gene expression data is
B-cell Acute Lymphoblastic Leukemia (ALL) data set, provided by
\cite{chiaretti2004gene} and maintained by \cite{RpkgALL}.
In ALL data set, we have $ X \in \mathbb{R}^{12,625} $ and
each observation belongs to one of the ordinal classes
$B_1 \prec B_2 \prec B_3 \prec B_4$ where
$ (n_1, n_2, n_3, n_4) = (19, 36, 23, 12) $.
The total sample size of the ALL data set is $ N = 90 $.
For these two data sets, we have compared the proposed methods with the competing methods discussed in Section \ref{subsec: competing}.
The analysis aims to identify up-regulated or down-regulated variables that are relevant to the ordinal classes while maintaining adequate classification accuracy.

We split the whole data set into train and test sets by 4:1, and then
further split the train set into the fitting and validation sets by 3:1.
That is, we randomly divide the data set into three parts with fitting:validation:test $=$ 3:1:1.
As in the sparse multiclass LDA literature, we first apply variable screening.
We chose variables by MV-SIS, a screening method for high dimensional classification problem \citep{cui2015model}, to rule out irrelevant variables at the training step.
MV-SIS scores are obtained from train set, and top 1,000 variables in Glioma data and 500 variables in ALL data are selected.
On the train set, we utilize the fitting and validation sets to tune and
fit the best model on the whole train set.
Finally, we measure losses on the test set and record the  number of
selected variables for each method.
This is repeated for 100 times, and Table \ref{tab: realdata} collects averaged
performance measures with standard error.

\begin{table}[t]
    \centering
    \resizebox{\columnwidth}{!}{%
    \begin{tabular}{cc|ccc|ccc|ccc|cc}
       & & \multicolumn{3}{c|}{MGSDA based} & \multicolumn{3}{c|}{fastPOI based} & \multicolumn{3}{c|}{MSDA based} & &  \\ \hline
    Data set & Metric & SOBL & OSBL & SLDA & SOBL & OSBL & SLDA & SOBL & OSBL & SLDA & FWOC & PCLM  \\
    \hline \hline
    \multirow{10}{*}{Glioma}
    & $l_0$
    &0.385& 0.491 &0.596&0.378& 0.382 &0.356  & 0.416  & 0.513 & 0.683 &0.328 &0.362\\
    & &(0.012)& (0.011) &(0.013)&(0.006)& (0.006) &(0.007)  & (0.014)  & (0.011) & (0.011) &(0.007) &(0.007)\\[5pt]
    & $l_1$
    &0.583 & 0.743 &1.02 &0.525& 0.522 &0.511   & 0.657 & 0.765 & 1.19 &0.456 &0.453\\
    & &(0.027)& (0.021) &(0.025)&(0.012)& (0.010) &(0.012)   & (0.034) & (0.022) & (0.023) &(0.013) &(0.009)\\[5pt]
    & $l_2$
    &1.07& 1.34 &2.09&0.866& 0.840 &0.899     & 1.27 & 1.37 & 2.51 &0.775 &0.648\\
    & &(0.070)& (0.053) &(0.061)&(0.031)& (0.024) &(0.031)     & (0.092) & (0.054) & (0.062) &(0.031) &(0.018)\\[5pt]
    & $|\widehat D|$
    &44.2& 153 &228.0&53.6& 45.2 &79.7              & 73.2 & 190.0 & 268.0 &334.0  &70.9\\
    & &(7.75)& (8.52) &(9.96)&(12.2)& (11.5) &(18.4)              & (9.16) & (7.69) & (9.63) &(33.7)  &(3.01)\\[5pt]
    & $|\widehat D \cap \widehat J_{ord}| / |\widehat D|$
    &0.968& 0.996 &0.930&0.956& 0.950 &0.475  & 0.982 & 0.995 & 0.912 &0.636&0.846\\
    & &(0.001)& (0.001) &(0.008)&(0.013)& (0.011) &(0.026)  & (0.003) & (0.001) & (0.004) &(0.031)&(0.005)\\
    \hline
    \multirow{10}{*}{ALL}
    &$l_0$
    &0.445 & 0.504 &0.434 &0.503 & 0.532 &0.508 & 0.464 & 0.517 & 0.451 &0.450 &0.476\\
    & &(0.012) & (0.011) &(0.011) &(0.012) & (0.011) &(0.010) & (0.011) & (0.010) & (0.011) &(0.011) &(0.021)\\[5pt]
    &$l_1$
    &0.555 & 0.618 &0.540 &0.625 & 0.649 &0.641 & 0.584 & 0.629 & 0.570 &0.561 &0.549\\
    & &(0.016) & (0.015) &(0.014) &(0.017) & (0.013) &(0.015) & (0.015) & (0.015) & (0.016) &(0.014) &(0.024)\\[5pt]
    &$l_2$
    &0.779 & 0.855 &0.764 &0.880 & 0.891 &0.928 & 0.827 & 0.866 & 0.824 &0.795 &0.694\\
    & &(0.029) & (0.028) &(0.026) &(0.033) & (0.025) &(0.032) & (0.026) & (0.030) & (0.032) &(0.028) &(0.037)\\[5pt]
    &$|\widehat D|$
    &18.7 & 7.32 &49.5 &14.7& 11.4 &101.0 & 28.6 & 11.9 & 75.6 &140.0 &39.9 \\
    & &(1.52) & (0.59) &(3.92) &(1.68) & (1.60) &(15.3)            & (1.60) & (0.613) & (4.22) &(14.5) &(2.9)\\[5pt]
    &$|\widehat D \cap \widehat J_{ord}|/|\widehat D|$
    &0.750 & 0.800 &0.178 &0.895 & 0.918 &0.244 & 0.795 & 0.813 & 0.212 &0.230 &0.395 \\
    & &(0.015) & (0.022) &(0.011) &(0.011) & (0.011) &(0.012) & (0.013) & (0.016) & (0.010) &(0.001) &(0.011)
    \end{tabular}%
    }
    \caption{Classification and variable selection performances of each method with
	100 repetitions on real data sets. Averaged measures are presented.
	Values in parentheses are standard errors.}
    \label{tab: realdata}
\end{table}

Denote the index set of sample ordinal variables as $\widehat J_{ord} := \{j: \hat{\mu}_{j}^1 \le \dots \le \hat{\mu}_{j}^K \mbox{ or } \hat{\mu}_{j}^1 \ge \dots \ge \hat{\mu}_{j}^K \}.$ To compare the variable selection performance, we report the number of selected variables denoted as $\widehat D$, and $|\widehat D \cap \widehat J_{ord}|/|\widehat D|$ which is the ratio of sample ordinal variables in selected variables.
In both datasets, regardless of the choice of base SLDA methods (MGSDA or fastPOI or MSDA), SOBL and OSBL select fewer variables than the corresponding sparse multiclass LDA methods while maintaining similar or lower losses than sparse LDA methods.
Furthermore, all SOBL and OSBL methods have high values of $|\widehat D \cap \widehat J_{ord}|/|\widehat D|$. Note that other competing methods, even the ordinal classification-oriented methods such as FWOC and PCLM selected many nominal variables than our proposed methods.

It is notable that in the case of the Glioma dataset, SOBL-fastPOI and OSBL-fastPOI report smaller values of $ l_2 $ loss than fastPOI even though fastPOI has a lower $ l_0 $ loss than SOBL-fastPOI and OSBL-fastPOI.
In Glioma dataset, FWOC performs best in terms of classification performance with respect to $ l_0 $ loss, but FWOC selects too many variables which makes interpretations somewhat clumsy. Similarly, while PCLM appears to be optimal in terms of the $ l_2 $ loss, it selects significantly more variables than SOBL-MGSDA and OSBL-fastPOI.  
Overall, SOBL methods report lower $ l_k $ losses than OSBL methods.

We take a further look at the difference between the selected variables of SOBL and other methods.
In both datasets, we fit SOBL-MGSDA, MGSDA, FWOC and PCLM. After fitting, we plot the sample group means of selected variables in Figure \ref{fig:varsel-realdata}.
The results from SOBL method based on other SLDAs and the OSBL are omitted as their patterns are similar to that of SOBL-MGSDA. 

\begin{figure}
\centering
\begin{subfigure}{\textwidth}
    \includegraphics[width=\textwidth]{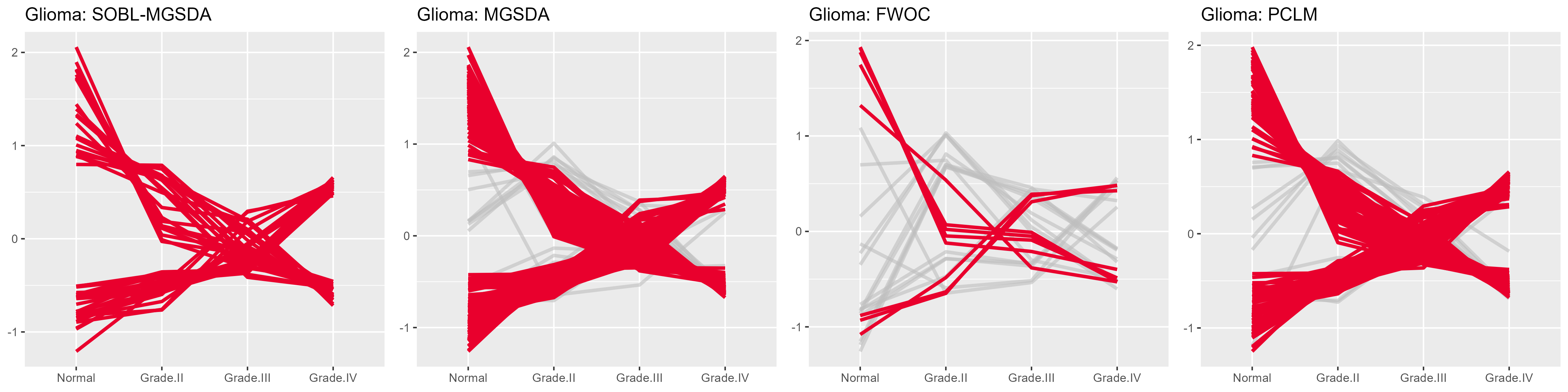}
\end{subfigure}

\begin{subfigure}{\textwidth}
    \includegraphics[width=\textwidth]{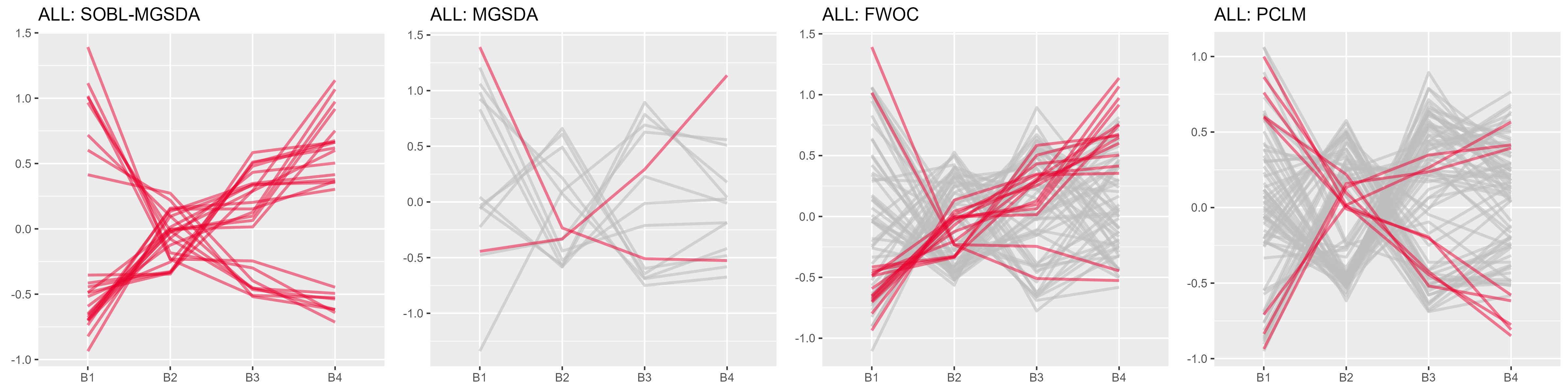}
\end{subfigure}
\caption{The class mean plots of selected variables in Glioma (top) and ALL (bottom) data sets.}
\label{fig:varsel-realdata}
\end{figure}

For the Glioma data, not only ordinal classification methods but also the SLDA method (MGSDA) has a high value of $|\widehat D \cap \widehat J_{ord}|/|\widehat D|$ in Table \ref{tab: realdata}, which implies that the Glioma data set has a strong ordinal signal. This is also observed in the first row of Figure \ref{fig:varsel-realdata}, in which most selected variables are order-concordant. 
The distinction between SOBL and the other methods is more clearly revealed in the class mean plot of the ALL dataset (bottom row). In the ALL dataset, nominal variables' signal strength is comparable to or even stronger than ordinal variables. Despite this, It is remarkable that only SOBL properly selects ordinal variables, while FWOC and PCLM do not.

The selected variables (genes) by SOBL are  all either upregulated or downregulated in the Glioma and ALL datasets. Many of the identified genes are experimentally known to have such relations. For example, the gene RIMS3 identified for the ALL dataset is known to be highly upregulated \citep{hicks2016molecular, malmberg2019accurate}, 
and KIAA0101 for the Glioma dataset is known to be related to progression of cancer \citep{wang2022kiaa0101}. 
For a detailed discussion, see Section S4 of the supplementary material.

\section{Discussion}\label{sec:discussion}

In this paper, we have discussed a novel basis learning and variable selection procedure for ordinal classification. Identifying both order-concordant and discriminatory variables provides better interpretability and explainability than simply pursuing classification accuracy. We believe that data analysts should decide whether to aim only at enhancing classification performance or at providing order-concordant and discriminatory variables. The proposed SOBL can be tuned to satisfy both purposes. This is shown analytically in Section \ref{subsec: theory}. We also have numerically demonstrated that SOBL with a greedy tuning can provide  classification accuracy at least comparable to, if not better than,  regular sparse LDA methods, and that SOBL with a two-step tuning procedure excels at selecting the order-concordant and discriminatory variables.

\clearpage
\appendix

\renewcommand{\thesection}{S\arabic{section}}
\renewcommand{\theequation}{S.\arabic{equation}}
\renewcommand{\thethm}{S\arabic{thm}}
\renewcommand{\thefigure}{S\arabic{figure}}

\begin{center}
    \textbf{\LARGE \MakeUppercase{Supplementary Material}}
\end{center}

\section{Technical details}
\subsection{Proof of Theorem 1}
We proceed the proof with proving the result for univariate case and then aggregate the results of all variables.
In this section, we illustrate tail probabilities of two types of Kendall's rank correlations $ \hat \tau $ and $ \tilde \tau $ introduced
in Section 3.3.

Recall that
\begin{enumerate}
	\item $ \hat \tau_j =  \frac{2}{N(N-1)}\sum_{1 \le i_1 < i_2 \le N} \mbox{sign}(X_{i_2j} - X_{i_1j}) \mbox{sign}(Y_{i_2} - Y_{i_1}) $,
	\item $ \tilde \tau_j  =
	\frac{2}{K(K-1)}\sum_{1 \le g_1 < g_2 \le K} \mbox{sign}(\hat \mu_j^{g_2} - \hat \mu_j^{g_1})$.
\end{enumerate}

\begin{lem} \label{lem: tau tail bound}
	For any $ \mu_{j}^{1}, \dots, \mu_j^{K} $, we have
	\begin{equation} \label{eq: conc tau hat}
		\mbox{Pr}\left( \vert \hat \tau_j - \mathbb E \hat \tau_j \vert \ge t	 \right)
		\le 2\exp\left( - \frac{Nt^2}{8} \right).
	\end{equation}
\end{lem}
\begin{proof}
	Suppose that the class membership $ \mathbf{Y} = (y_1, \dots, y_N) $ are given.
	For any $ i = 1, \dots, N $ and for any $ x_1, \dots, x_N, x_i' \in \mathbb R $,
	Then we have
	\begin{align*}
		&\vert \hat \tau_j(x_1, \dots, x_{i-1}, x_i,  x_{i+1}, \dots, x_N) -
		\hat \tau_j(x_1, \dots, x_{i-1}, x_i',  x_{i+1}, \dots, x_N) \vert \\
		&\le \frac{2}{N(N-1)} \sum_{k \neq i} \vert \mbox{sign}(x_i - x_k) - \mbox{sign}(x_i' - x_k) \vert  \\
		&\le \frac{2}{N(N-1)} \cdot 2(N-1) = \frac{4}{N}.
	\end{align*}
	By McDiamard's inequality, 
    we have 
	\begin{equation*} 
		\mbox{Pr}\left( \vert \hat \tau_j - \mathbb E \hat \tau_j \vert \ge t	 
		~|~ \mathbf Y \right)
		\le 2\exp\left( - \frac{Nt^2}{8} \right), 
	\end{equation*}
	and by taking expectation on $ \mathbf Y $ on both sides, we have \eqref{eq: conc tau hat}.
\end{proof}

\begin{lem} \label{lem: tau-o}
	Suppose that $ \mu^1_j, \dots, \mu^K_j $ are all different.
	Also, assume that $ c_1 \le \pi_g \le c_2 $ for some fixed constants $ c_1, c_2 > 0 $ for all $ g = 1, \dots, K $. 
	Let
	$$ \tau^o_j = \frac{2}{K(K-1)} \sum_{1 \le g_1 < g_2 \le K} I(\mu^{g_1}_j < \mu^{g_2}_j) - I(\mu^{g_1}_j > \mu^{g_2}_j) $$
	and
	$ \vartheta_j = \min_{1 \le g_1 < g_2 \le K} = \frac{|\mu_j^{g_2} - \mu_j^{g_1}|}{\sigma} .$
	Then
	$$
	\vert
	\mathbb{E} \tilde{\tau}_j - \tau^o_j
	\vert
	\le
	8\exp(-CN(\vartheta_j^2 \land 1)).
	$$
\end{lem}
\begin{proof}
	Fix some indices $ 1 \le g_1 < g_2 \le K $. 
	Note that for $ t > 0 $, 
	$$
		\mbox{Pr}(\hat \mu_j^{g_2} - \hat \mu_j^{g_1} - (\mu_j^{g_2} - \mu_j^{g_1}) > t)
		\le \mbox{Pr}\left(\hat \mu_j^{g_2} - \mu_j^{g_2} > \frac{t}{2}\right)
		+ \mbox{Pr}\left(\hat \mu_j^{g_1} - \mu_j^{g_1} < -\frac{t}{2}\right).
	$$
	By Hoeffding's inequality (Theorem 2.6.3 of \cite{vershynin2018high}),
	$$
		\mbox{Pr}\left(\hat \mu_j^{g_2} - \mu_j^{g_2} > \frac{t}{2} ~\bigg|~ \mathbf{Y}\right)
		\le \exp\left(-\frac{C n_{g_2} t^2}{\sigma^2}\right).
	$$
	By applying same argument from the proof of Proposition 1 in the supplementary materials of \cite{mai2019multiclass}, we get 
	\begin{equation*}
		\mbox{Pr}\left(\hat \mu_j^{g_2} - \mu_j^{g_2} > \frac{t}{2}\right)
		\le
		\mathbb E \left[\exp\left(-\frac{C n_{g_2} t^2}{\sigma^2}\right)\right]
		\le \exp(-CNt^2/\sigma^2) + \exp(-CN).
	\end{equation*}
	With the same process, $\mbox{Pr}\left(\hat \mu_j^{g_1} - \mu_j^{g_1} < -t/2\right)  \le \exp(-CNt^2/\sigma^2) + \exp(-CN) $ holds.
	Thus, we have 
	\begin{equation}\label{eq:mean-diff-concentration}
		\mbox{Pr}(\hat \mu_j^{g_2} - \hat \mu_j^{g_1} - (\mu_j^{g_2} - \mu_j^{g_1}) > t)
		\le 
		2\left(\exp(-\frac{CNt^2}{\sigma^2}) + \exp(-CN)\right)
	\end{equation}
    Similarly, we also have $\mbox{Pr}(\hat \mu_j^{g_2} - \hat \mu_j^{g_1} - (\mu_j^{g_2} - \mu_j^{g_1}) < -t)
		\le 
		2\left(\exp(-\frac{CNt^2}{\sigma^2}) + \exp(-CN)\right)$
    \begin{equation}\label{eq:mean-diff-concentration-negative}
		\mbox{Pr}(\hat \mu_j^{g_2} - \hat \mu_j^{g_1} - (\mu_j^{g_2} - \mu_j^{g_1}) < -t)
		\le 
		2\left(\exp(-\frac{CNt^2}{\sigma^2}) + \exp(-CN)\right)
	\end{equation}

	Now, suppose that $ \mu^{g_2}_{j} > \mu^{g_1}_{j} $.
	By \eqref{eq:mean-diff-concentration-negative}, we have
	\begin{align}
		|\mathbb E \mbox{sign}(\hat \mu^{g_2}_j - \hat \mu^{g_1}_j) - 1|
		&= 2|\mathbb P (\hat \mu^{g_2}_j - \hat \mu^{g_1}_j > 0 ) - 1 | \nonumber\\ 
		&= 2\mathbb P(\hat \mu^{g_2}_j - \hat \mu^{g_1}_j - (\mu^{g_2}_j - \mu^{g_1}_j) < -(\mu^{g_2}_j - \mu^{g_1}_j)) \nonumber \\
		&\le 4\left(\exp(-CN(\mu_j^{g_2} - \mu_j^{g_1})^2/\sigma^2 ) 
		+ \exp(-CN)\right) \nonumber \\ 
		&\le 8\exp(-CN(\vartheta_j^2 \land 1)). \label{eq: bd-Wg12}
	\end{align}
    In the case of $ \mu_j < \mu_i $, we have the same bound for $|\mathbb E \mbox{sign}(\hat \mu_j - \hat \mu_i) + 1|$, i.e., 
    \[
    |\mathbb E \mbox{sign}(\hat \mu_j - \hat \mu_i) + 1| \le 8\exp(-CN(\vartheta_j^2 \land 1)).
    \]
	Finally, we have that
	\begin{align*}
		\vert \mathbb{E} \tilde{\tau_j} - \tau^o_j \vert
		&\le \frac{2}{K(K-1)} \left\{
		\sum_{\substack{g_1 < g_2, \\ \mu^{g_1}_j < \mu^{g_2}_j}} \vert \mathbb E \mbox{sign}(\hat \mu^{g_2}_j - \hat \mu^{g_1}_j) - 1 \vert	+
		\sum_{\substack{g_1 < g_2, \\ \mu^{g_1}_j > \mu^{g_2}_j}} \vert \mathbb E \mbox{sign}(\hat \mu^{g_2}_j - \hat \mu^{g_1}_j) + 1 \vert
		\right\} \\
		&\le 8\exp(-CN(\vartheta_j^2 \land 1)).
	\end{align*}
\end{proof}

\begin{lem} \label{lem: tilde tau conc}
	Assume the same setting of Lemma \ref{lem: tau-o}. Then for any $ t > 0 $,
	$$
	\mbox{Pr}\left( \vert \tilde{\tau}_j - \mathbb{E} \tilde{\tau}_j \vert > t \right)
	\le
	6K^2\exp(-CN(\vartheta_j^2 \land 1)).
	$$
\end{lem}
\begin{proof}
	Fix some $ 1 \le g_1 < g_2 \le K $ and denote $ W_{j}^{g_1, g_2} = \mbox{sign}(\hat \mu_j^{g_2} - \hat \mu_j^{g_1}) $ for a notational convenience. 

	Suppose that $\mu_{j}^{g_1} < \mu_{j}^{g_2}.$
	Then
	\begin{align*}
	&\mbox{Pr}(|W_{j}^{g_1, g_2} - \mathbb E W_{j}^{g_1, g_2}| > t) \\
	&= \mbox{Pr}(|W_{j}^{g_1, g_2} - \mathbb E W_{j}^{g_1, g_2}| > t,~\hat \mu_{j}^{g_2} > \hat \mu_{j}^{g_1})
	+ \mbox{Pr}(|W_{j}^{g_1, g_2} - \mathbb E W_{j}^{g_1, g_2}| > t,~\hat \mu_{j}^{g_2} < \hat \mu_{j}^{g_1}) \\
	&\le \mbox{Pr}(|\mathbb E W_{j}^{g_1, g_2} - 1| > t,~\hat \mu_{j}^{g_2} > \hat \mu_{j}^{g_1})
	+ \mbox{Pr}(\hat \mu_{j}^{g_2} < \hat \mu_{j}^{g_1}) \\
	&\le I(t \le 8e^{-CN(\vartheta_j^2 \land 1)}) + \mbox{Pr}(\hat \mu_{j}^{g_2} - \hat \mu_{j}^{g_1} - (\mu_{j}^{g_2} - \mu_{j}^{g_1}) < - (\mu_{j}^{g_2} - \mu_{j}^{g_1})) \\
	&\le 2\exp(-e^{CN( \vartheta_j^2 \land 1)})
	+ 4\exp(-CN(\vartheta_j^2 \land 1)) \\
	&\le 6\exp(-CN(\vartheta_j^2 \land 1)).
	\end{align*}
        The fourth line inequality comes from \eqref{eq: bd-Wg12}, and the fifth line uses \eqref{eq:mean-diff-concentration-negative}.
	We have the same upper bound for the case of $ \mu_{j}^{g_1} > \mu_{j}^{g_2} $ by similar arguments.
	Putting together, we have 
	\begin{align*}
		\mathbb P(|\tilde \tau_j - \mathbb E \tilde \tau_j| > t) 
		&\le \mathbb P \left(
			\frac{2}{K(K-1)}
			\sum_{1 \le g_1 < g_2 \le K}
			\left \vert
			W_{j}^{g_1, g_2} - \mathbb E W_{j}^{g_1, g_2}
			\right \vert > t
		\right) \\
		&\le \mathbb P \left(
			\bigcup_{1 \le g_1 < g_2 \le K}
			\left\{
			\left \vert
			W_{j}^{g_1, g_2} - \mathbb E W_{j}^{g_1, g_2}
			\right \vert > t
			\right\}
		\right) \\
		&\le K^2 \max_{1 \le g_1 < g_2 \le K}
		\mathbb P \left(
			\left \vert
			W_{j}^{g_1, g_2} - \mathbb E W_{j}^{g_1, g_2}
			\right \vert > t
		\right) \\
		&\le 6K^2\exp(-CN(\vartheta_j^2 \land 1)). 
	\end{align*}
\end{proof}

Now we are ready to prove Theorem 1.

\begin{proof}[Proof of Theorem 1]
	Recall that the ordinal weight $ w_j $ of the $ j $th variable is given as
	$$
	w_j = I\left(\vert \hat \tau_j \vert > \frac{\Delta}{2}\right) I\left(\vert \tilde \tau_j \vert > 1-\frac{2}{K(K-1)}\right).
	$$
    Suppose that $ j \in J_{noise} $. By Condition (C2), $X_{i_2 j} - X_{i_1 j} \overset{d}{=} -(X_{i_2 j} - X_{i_1 j}) ~|~ \mathbf{Y} $ holds, and hence $ \mathbb E \left[\mbox{sign}(X_{i_2 j} - X_{i_1 j}) | \mathbf Y \right] = 0 $.
	As a consequence, we have 
	$ \mathbb E \hat \tau_j = \mathbb E [\mathbb E [\hat\tau_j | \mathbf Y]] = 0 $.
	Then Lemma \ref{lem: tau tail bound} implies that
	\begin{equation} \label{eq: proof-thm1-1}
	\mbox{Pr}\left( \vert \hat \tau_j \vert \ge \Delta / 2 \right) \le
	2\exp\left( -\frac{N\Delta^2}{32} \right).
	\end{equation}
	On the other hand, if $ j \in J_{md} $ then $\mbox{Pr}\left( \vert \hat \tau_j \vert \ge \Delta / 2  \right)
	\ge \mbox{Pr}\left( \vert \hat \tau_j - \mathbb{E} \hat \tau_j \vert < \Delta/2 \right)$ since  $|\hat \tau_j - \mathbb E \hat \tau_j| \ge |\mathbb E \hat \tau_j| - |\hat \tau_j| \ge \Delta - |\hat \tau_j|$.
        By Lemma \ref{lem: tau tail bound}, we have
	\begin{equation} \label{eq: proof-thm1-2}
	\mbox{Pr}\left( \vert \hat \tau_j \vert \ge \Delta / 2  \right)
	\ge
	1 - 2\exp\left(  -\frac{N\Delta^2}{32} \right).	
	\end{equation}
	Define the event
	$ A_j  = \{\vert \hat \tau_j \vert < \Delta / 2\}$ for $ j = 1, \dots, p $
	and let
	$ \mathcal{A} \coloneqq \left( \cap_{j \in J_{md}} A_j^c \right) \cap
	\left( \cap_{j \in J_{noise}} A_j \right) $.
	Then (\ref{eq: proof-thm1-1}) and (\ref{eq: proof-thm1-2}) imply
	\begin{align}
		\mbox{Pr}\left( \mathcal{A} \right)
		&\ge 1 - \sum_{j \in J_{md}} \mbox{Pr}\left( A_j\right) - \sum_{j\in J_{noise}}\mbox{Pr}\left(A_j^c\right) \nonumber \\
		&\ge 1 - 2|J_{md}|\exp\left(-\frac{N\Delta^2}{32}\right) - 2(p-|J_{md}|)\exp\left(-\frac{N\Delta^2}{32}\right) \nonumber \\
		&= 1 - 2p\exp\left(-\frac{N\Delta^2}{32}\right). \label{eq: pr Anoise}
	\end{align}
	
	From now on, we assume that $ \mathcal{A} $ is occurred.
	Clearly, $ w_j = 0 $ if $ j \notin J_{md} $ under the event $ \mathcal{A} $.
	Now, suppose that $ j \in J_{md} $.
	We note that $ \vert \tau^o_j \vert = 1 $ if
	$ \mu_{j}^{1} < \dots < \mu_{j}^{K} $ or $ \mu_{j}^{1} > \dots > \mu_{j}^{K} $,
	which holds for any $ j \in J_{ord} $ since $ \vartheta_{\min} > 0 $ by assumption.
	Otherwise, if $ j \in J_{nom} $ then
	\begin{equation} \label{eq: proof-thm1-5}
		\vert \tau^o_j \vert \le 1 - \frac{4}{K(K-1)}.
	\end{equation}
	The maximum is achieved when group means are all monotone except that only
	one consecutive pair has the opposite direction (see Section 3.4).
	By Lemma 8, for the $j$th variable, we have
	\begin{equation} \label{eq: proof-thm1-3}
	\vert \mathbb{E} \tilde{\tau}_j - \tau^o_j \vert
	\le 8\exp(-CN(\vartheta_j^2 \land 1))
	\le 8\exp(-CN(\vartheta_{\min}^2 \land 1)).
	\end{equation}
	We may choose sufficiently large $ N $ such that
	\begin{equation} \label{eq: proof-thm1-4}
	\epsilon_{N,  \vartheta_{\min}} \coloneqq
	8\exp(-CN(\vartheta_{\min}^2 \land 1))
	< \frac{1}{K(K-1)}.
	\end{equation}
	Now suppose that $ \vert \tilde \tau_j - \mathbb{E} \tilde \tau_j \vert \le \frac{1}{K(K-1)}$.
	If $ j \in J_{ord} $, by (\ref{eq: proof-thm1-3}), (\ref{eq: proof-thm1-4})
	and the triangle inequality, we have
	$$
	\vert \mathbb{E}\tilde \tau_j \vert
	\ge \vert \tau_j^o \vert - |\mathbb E \tilde \tau_j - \tau_j^o|
	\ge \vert \tau_j^o \vert - \epsilon_{N,  \vartheta_{\min}}
	\ge 1 - \frac{1}{K(K-1)},
	$$
	and
	$$
	\vert \tilde \tau_j \vert
	\ge \vert \mathbb{E} \tilde \tau_j \vert - |\tilde \tau_j - \mathbb E \tilde \tau_j|
	\ge \vert \mathbb{E} \tilde \tau_j \vert - \frac{1}{K(K-1)}
	\ge 1 - \frac{2}{K(K-1)}.
	$$
	Otherwise, if $ j \in J_{nom} $ then, by (\ref{eq: proof-thm1-5}),
	(\ref{eq: proof-thm1-3}) and (\ref{eq: proof-thm1-4}) and the triangle inequality,
	$$
	\vert \mathbb{E}\tilde \tau_j \vert
	\le \vert \tau_j^o \vert + \epsilon_{N,  \vartheta_{\min}}
	\le 1 - \frac{3}{K(K-1)},	
	$$
	and
	$$
	\vert \tilde \tau_j \vert
	\le \vert \mathbb{E} \tilde \tau_j \vert + \frac{1}{K(K-1)}
	\le 1 - \frac{2}{K(K-1)}.
	$$
	Define the event
	$ B_j = \{\vert \tilde{\tau}_j - \mathbb E \tilde \tau_j \vert \le \frac{1}{K(K-1)}\} $
	for $ j = 1, \dots, p $
	and denote $ \mathcal{B} = \cap_{1 \le j \le p} B_j $.
	Then under the event $ \mathcal A \cap \mathcal B $, we have
	$ \vert \tilde \tau_j \vert \ge 1 - \frac{2}{K(K-1)}$ if $ j \in J_{ord} $ and
	$ \vert \tilde \tau_j \vert \le 1 - \frac{2}{K(K-1)}$ if $ j \in J_{nom} $.
	By Lemma \ref{lem: tilde tau conc}, we have that
	\begin{align*}
        \mbox{Pr}\left(\mathcal B \right)
		\ge
		1 - \sum_{j=1}^p \mbox{Pr}\left( B_j^c \right)
		\ge
		1 - p \cdot 6K^2\exp(-CN(\vartheta_{\min}^2 \land 1)).
	\end{align*}
	Aggregating this bound and (\ref{eq: pr Anoise}) completes the proof.
\end{proof}

Let  $ \mathcal E_{ord} = \mathcal A \cap \mathcal B $ where $ \mathcal A $ and $ \mathcal B $ are events defined in the proof of Theorem 1. We note that under $ \mathcal E_{ord} $, $ w_j = 1 $ if $ j \in J_{ord} $ and $ w_j = 0 $ otherwise.

\subsection{Proofs for Section 4} \label{appendix:pfsec4}

\subsubsection{Technical lemmas}

The following lemma implicitly comes from \cite{mai2012direct}.
\begin{lem} \label{lemma: dividing}
	For $ V, \Delta \in \mathbb{R}^{p \times d} $ and $ \Upsilon \in \mathbb{R}^{p \times p} $, suppose that
	$ V = \Upsilon^{-1}\Delta $.
	Let $ A = \{i: V_{i,:} \neq 0\} $, i.e., $ V = \begin{bmatrix}
		V_{A,:} \\ 0
	\end{bmatrix}. $
	Then,
	\begin{enumerate}
		\item $ \Delta_{A^c, :} = \Upsilon_{A^c A} \Upsilon_{A A}^{-1} \Delta_{A, :} $
		\item $ V_{A,:} = \Upsilon_{A A}^{-1} \Delta_{A, :}$
	\end{enumerate}
	(Here, $ \Upsilon_{A A}^{-1} = (\Upsilon_{A A})^{-1} .$)
\end{lem}
\begin{proof}
	By applying the block matrix inversion formula, we get
	\[
	\Upsilon^{-1} =
	\begin{bmatrix}
		\Upsilon^{-1}_{AA} + \Upsilon^{-1}_{AA} \Upsilon_{AA^c} S^{-1} \Upsilon_{A^cA}
		\Upsilon^{-1}_{AA} & -\Upsilon^{-1}_{AA} \Upsilon_{AA^c} S^{-1} \\
		- S^{-1}\Upsilon_{A^cA}\Upsilon^{-1}_{AA}   &  S^{-1}
	\end{bmatrix},
	\]
	where $ S = \Upsilon_{A^cA^c} - \Upsilon_{A^cA} \Upsilon^{-1}_{AA}\Upsilon_{AA^c} $.
	Then,
	\[
	\begin{bmatrix}
		V_{A,:} \\
		0
	\end{bmatrix}
	=
	\begin{bmatrix}
		(\Upsilon^{-1}_{AA} + \Upsilon^{-1}_{AA} \Upsilon_{AA^c} S^{-1} \Upsilon_{A^cA}
		\Upsilon_{AA}^{-1})\Delta_{A,:}
		- \Upsilon^{-1}_{AA} \Upsilon_{AA^c} S^{-1} \Delta_{A^c, :} \\
		-S^{-1}\Upsilon_{A^cA}\Upsilon^{-1}_{AA} \Delta_{A,:} + S^{-1} \Delta_{A^c,:}	
	\end{bmatrix}.
	\]
	From the second component, the first part of the proposition is proved.
	Observe that
	\begin{equation*}
		V_{A,:}
		= \Upsilon^{-1}_{AA}\Delta_{A,:}  +
		\Upsilon^{-1}_{AA} \Upsilon_{AA^c} S^{-1} (\Upsilon_{A^cA}\Upsilon^{-1}_{AA}\Delta_{A,:} - \Delta_{A^c,:})
		= \Upsilon^{-1}_{AA}\Delta_{A,:}.
	\end{equation*}
	The last equality holds by the first part of the lemma.
\end{proof}

\begin{lem} \label{lem:norms}
	Let $A, B$ be multiplicative matrices. Then
	\begin{enumerate}
		\item $ \Vert A \Vert_{\infty, 2} \le \Vert A \Vert_{\infty} $,
		\item $ \Vert AB \Vert_{\infty, 2} \le \Vert A \Vert_{\infty} \Vert B \Vert_{\infty, 2} $
	\end{enumerate}
\end{lem}
\begin{proof}
	For any vector $ v $, it holds that $ \|v\|_2 \le \|v\|_1 $.
	From the definitions,
	$$
		\|A\|_{\infty, 2} = \max_i \|A_{i, :}\|_2 \le \max_i \|A_{i, :}\|_1 = \|A\|_{\infty}.
	$$

	For the second part,
	see the proof of Lemma 1 in Section S2.1 of the supplementary material of
	\cite{gaynanova2016sparseLDA}.
\end{proof}
	
	For the proofs, we only consider the case of SOBL based on MSDA \citep{mai2019multiclass}.
	In the case of MGSDA \citep{gaynanova2016sparseLDA}, it is straightforward to proceed proofs
	by using similar arguments.
	Thus, we let $ \widehat \Sigma = \widehat \Sigma_w $ and
	$ \widehat M = [\hat \mu_2 - \hat \mu_1, \dots, \hat \mu_K - \hat \mu_1] $.
	Also, we denote $ A \coloneqq J_{disc} $,
	$ A_1 \coloneqq J_{disc}^{ord} = J_{disc} \cap J_{ord} $,
	$ A_2 \coloneqq J_{disc} \cap J_{ord}^c $,
	$ A_3 \coloneqq J_{disc}^c \cap J_{ord} $, and
	$ A_4 \coloneqq J_{disc}^c \cap J_{ord}^c$ for
	notational simplicity.
	
	First, we introduce some concentration results.
	Let $ \epsilon > 0 $.
	Two event sets  $\mathcal E_{M}(\epsilon)$ and $\mathcal E_{\Sigma}(\epsilon)$ are defined by
	\begin{align*}
		\mathcal E_{M}(\epsilon) &= \{|\widehat M_{ij} - M_{ij}| < \frac{\epsilon}{d} ~
		\mbox{ for any } i =1, \dots, p \mbox{ and } j=1, \dots, K-1\}, \\
		\mathcal E_{\Sigma}(\epsilon) &= \{|\widehat \Sigma_{ij} - \Sigma_{ij}| < \frac{\epsilon}{d} ~
		\mbox{ for any } i = 1, \dots, p \mbox{ and } j \in A\},
	\end{align*}
	where $ d = \vert A \vert $.
	The following lemma provides lower bound of probability that event sets
	$ \mathcal E_{\Sigma}(\epsilon) $ and $ \mathcal E_{M}(\epsilon) $ occur.
	\begin{lem} \label{lem: event sets}
		For any $ \epsilon > 0$, we have
		$$
		\Pr(\mathcal E_{M}(\epsilon)\cap \mathcal E_{\Sigma}(\epsilon)) \ge 1 - CKp e^{-\frac{CN\epsilon^2}{d^2\sigma^2}} - CKpd e^{-C\min (\frac{\epsilon^2}{d^2\sigma^4 K^2}, \frac{\epsilon}{d \sigma^2 K})N}.
		$$
		Here, $C > 0$ is a generic constant.
	\end{lem}
	\begin{proof}
		From the proof of Lemma \ref{lem: tau-o}, for any $ t > 0 $, 
		we have
		$$
			\mbox{Pr}(\vert\widehat{M}_{ij} - M_{ij}\vert > t)
			\le 4\left(\exp(-\frac{CNt^2}{\sigma^2}) + \exp(-CN)\right).
		$$
		
		Next, we will bound the probability of $ \{|\widehat \Sigma_{ij} - \Sigma_{ij}| > t\} $.
		Observe that 
		$$
			\widehat{\Sigma}_{ij}
			= \frac{1}{N}\sum_{g=1}^K \sum_{Y_l = g} (X_{li} - \mu_i^g)(X_{lj} - \mu_j^g)
			- \sum_{g=1}^K \frac{n_g}{N} (\hat \mu_i^g - \mu_i^g)(\hat \mu_j^g - \mu_j^g) =: I_1 - I_2.
		$$
		By the Hanson-Wright inequality for a bilinear form (Theorem 1 of \cite{park2023sparse}),
		\begin{align*}
			\Pr\left(
				\bigg\vert 
					\frac{1}{n_g} \sum_{Y_l = g}
					(X_{li} - \mu_i^g)(X_{lj} - \mu_j^g) - \Sigma_{ij}
				\bigg\vert
				> t ~\bigg\vert~
				\mathbf{Y}
			\right)
			\le 2\exp\left(-C n_g \min\left(\frac{t^2}{\sigma^4}, \frac{t}{\sigma^2}\right)\right).
		\end{align*}
		Thus, 
		\begin{align*}
			&\Pr(|I_1 - \Sigma_{ij}| > t) \\
			&= \Pr\left(
				\left|\sum_{g=1}^K \frac{n_g}{N} \left(\frac{1}{n_g} \sum_{Y_l = g} (X_{li} - \mu_i^g)(X_{lj} - \mu_j^g) - \Sigma_{ij}\right) \right| > t
			\right) \\
			&= \mathbb E \left(
				\Pr\left(
					\bigg\vert\sum_{g=1}^K \frac{n_g}{N} (\frac{1}{n_g} \sum_{Y_l = g} (X_{li} - \mu_i^g)(X_{lj} - \mu_j^g) - \Sigma_{ij}) \bigg\vert > t
				 ~\bigg|~ \mathbf Y \right)
			\right) \\ 
			&\le 
			\sum_{g=1}^K \mathbb E \left(
				\Pr\left(\bigg|\frac{1}{n_g} \sum_{Y_l = g} (X_{li} - \mu_i^g)(X_{lj} - \mu_j^g) - \Sigma_{ij}\bigg| > \frac{Nt}{K n_g} ~\bigg|~ \mathbf{Y}
				\right)
			\right) \\
			&\le \sum_{g=1}^K \mathbb E \left(
				\Pr\left(\bigg|\frac{1}{n_g} \sum_{Y_l = g} (X_{li} - \mu_i^g)(X_{lj} - \mu_j^g) - \Sigma_{ij}\bigg| > \frac{t}{K} ~\bigg|~ \mathbf{Y}
				\right) 
			\right) \\ 
			&\le \sum_{g=1}^K \mathbb E 
			\left(
				2\exp\left(-C n_g \min\left(\frac{t^2}{\sigma^4 K^2}, \frac{t}{\sigma^2 K}\right)\right)
			\right)
			\\
			&\le 
			4K e^{-C\min (\frac{t^2}{\sigma^4 K^2}, \frac{t}{\sigma^2 K})N}.
		\end{align*}
		For $ I_2 $, we have
		\begin{align*}
			&\Pr\left( \left| \frac{n_g}{N}(\hat \mu_i^g - \mu_i^g)(\hat \mu_j^g - \mu_j^g) \right| > \frac{t}{K} ~\bigg\vert~ \mathbf{Y} \right) \\
			&\le \Pr\left((\hat \mu_i^g - \mu_i^g)^2 + (\hat \mu_j^g - \mu_j^g)^2 > \frac{2Nt}{n_g K} ~\bigg\vert~ \mathbf{Y} \right) \\
			&\le \Pr\left(|\hat \mu_i^g - \mu_i^g| > \sqrt \frac{Nt}{n_g K}~\bigg\vert~ \mathbf{Y} \right)
			+ \Pr\left(|\hat \mu_j^g - \mu_j^g| > \sqrt \frac{Nt}{n_g K}~\bigg\vert~ \mathbf{Y} \right) \\
			&\le 4\exp\left(- \frac{ C N t}{K \sigma^2}\right).
		\end{align*}
		Then  we get 
		$$
			\Pr(|I_2| > t)
			\le \sum_{g=1}^K \Pr\left( \left| \frac{n_g}{N}(\hat \mu_i^g - \mu_i^g)(\hat \mu_j^g - \mu_j^g) \right| > \frac{t}{K}\right)
			\le 4K\exp\left(- \frac{ C N t}{K \sigma^2}\right).
		$$
		From these bounds, we have
		$$
			\Pr(|\widehat \Sigma_{ij} - \Sigma_{ij}| > t)
			\le \Pr(|I_1| > t/2) + \Pr(|I_2| > t/2)
			\le 8K e^{-C\min (\frac{t^2}{\sigma^4 K^2}, \frac{t}{\sigma^2 K})N}.
		$$
		By the union bounds, we get the result.

	\end{proof}

	\begin{lem} \label{lem: concentrations}
		Assume that both $ \mathcal E_{\Sigma}(\epsilon) $ and $ \mathcal E_{M}(\epsilon) $ defined in Lemma \ref{lem: event sets} have occurred.
		\begin{enumerate}
			\item We have
			\begin{align*}
                \max\{\|\widehat \Sigma_{AA} - \Sigma_{AA} \|_{\infty}, \|\widehat \Sigma_{A^cA} - \Sigma_{A^cA} \|_{\infty}, \|\widehat M - M\|_{\infty}\} < \epsilon.
			\end{align*}
			
			\item For $ \epsilon < \frac{1}{\phi} $, we have
			\begin{align*}
				\|(\widehat{\Sigma}_{AA})^{-1} - (\Sigma_{A A})^{-1}\|_{\infty} &< \frac{\epsilon\phi^2}{1-\phi\epsilon}; \\
				\|\widehat{\Sigma}_{A^c A}(\widehat{\Sigma}_{AA})^{-1} - \Sigma_{A^c A}(\Sigma_{A A})^{-1}\|_{\infty} &< \frac{\epsilon\phi(1 + \kappa)}{1-\phi\epsilon}.
			\end{align*}
			
			\item Let $ \phi_1 = \|(\Sigma_{A_1 A_1})^{-1}\|_{\infty} $ and
			$ \kappa_i =  \|\Sigma_{A_i A_1}(\Sigma_{A_1 A_1})^{-1}\|_{\infty}$ for $ i = 2, 3, 4 $.
			For $ \epsilon < \frac{1}{\phi_1} $, we have
			\begin{align*}
				\|(\widehat{\Sigma}_{A_1 A_1})^{-1} - (\Sigma_{A_1 A_1})^{-1}\|_{\infty} &< \frac{\epsilon\phi_1^2}{1-\phi_1\epsilon}; \\
				\|\widehat{\Sigma}_{A_i A_1}(\widehat{\Sigma}_{A_1 A_1})^{-1} - \Sigma_{A_i A_1}(\Sigma_{A_1 A_1})^{-1}\|_{\infty} &< \frac{\epsilon\phi_1(1 + \kappa_i)}{1-\phi_1\epsilon},
				\quad \mbox{for }~ i=2, 3, 4.
			\end{align*}
		\end{enumerate}
		\begin{proof}[Proof of Lemma 14]
			The first part of the lemma is a natural consequence of Lemma \ref{lem: event sets}.
			The second and third part can be proved by following the proof of Lemma A2 in \cite{mai2012direct}.
		\end{proof}
	\end{lem}

	In the proofs of Theorems 2 and 3, we first assume $ \mathcal E_{\Sigma}(\epsilon) \cap \mathcal E_{M}(\epsilon) \cap \mathcal E_{ord} $ for an $\epsilon > 0$ and determine the $\epsilon$ to satisfy KKT conditions.  We follow the proof strategy of \cite{gaynanova2016sparseLDA}.

	\subsubsection{Proof of Theorem 2}
       
	\begin{proof}[Proof of 1.]
        For a $\epsilon > 0$, assume the event $ \mathcal E_{\Sigma}(\epsilon) \cap \mathcal E_{M}(\epsilon) \cap \mathcal E_{ord} $.
		From the KKT conditions, $ \widehat Z = [\widehat Z_{A_1}^\top, \widehat Z_{A_2}^\top,  0_{p-d}^\top]^\top $
		becomes the solution of the optimization problem (8) if
		\begin{align}
			\widehat\Sigma_{A_1, :} \widehat Z - \widehat M_{A_1} + \lambda \hat s_{A_1} &= 0, \label{eq: Thm5-1 KKT1} \\
			\widehat\Sigma_{A_2, :} \widehat Z - \widehat M_{A_2} + \lambda\eta \hat s_{A_2} &= 0, \label{eq: Thm5-1 KKT2} \\
			\|\widehat \Sigma_{A_3, :} \widehat Z - \widehat M_{A_3} \|_{\infty, 2} &\le \lambda, \label{eq: Thm5-1 KKT3} \\
			\|\widehat \Sigma_{A_4, :} \widehat Z - \widehat M_{A_4} \|_{\infty, 2} &\le \lambda \eta \label{eq: Thm5-1 KKT4}
		\end{align}	
		where $\hat s = \begin{bmatrix} \hat s_{A_1} \\ \hat s_{A_2}\end{bmatrix} $
		is a subgradient of $ \sum_{j=1}^p \|\widetilde Z_j\|_2 $ at
		$ \widehat Z $ such that for each $ j \in J $,
		$$
		s_{j, :} = \begin{cases}
			\frac{\widetilde Z_j}{\|\widetilde Z_j\|_2}, &\mbox{if }~ \widetilde Z_j \neq 0; \\
			\in \{v \in \mathbb R^{K-1} : \|v\|_2 \le 1 \}, &\mbox{if }~ \widetilde Z_j = 0.
		\end{cases}	
		$$
		Let $ \tilde \Psi_A = (\widehat \Sigma_{AA})^{-1}\widehat M_A $.
		If $ \widehat Z $ satisfies (\ref{eq: Thm5-1 KKT1}) and (\ref{eq: Thm5-1 KKT2}) then we can rewrite $ \widehat Z $ by
		\begin{align*}
			\widehat Z_{A_1}
			&= \tilde \Psi_{A_1} - \lambda (\widehat \Sigma_{AA})^{-1}_{A_1 A_1} \hat s_{A_1}
			- \lambda\eta (\widehat \Sigma_{AA})^{-1}_{A_1 A_2} \hat s_{A_2}, \\
			\widehat Z_{A_2}
			&= \tilde \Psi_{A_2} - \lambda (\widehat \Sigma_{AA})^{-1}_{A_2 A_1} \hat s_{A_1}
			- \lambda\eta (\widehat \Sigma_{AA})^{-1}_{A_2 A_2} \hat s_{A_2}.
		\end{align*}
		Here, $ (\widehat{\Sigma}_{AA})^{-1}_{A_1 A_1} $ denotes
		$ (A_1, A_1) $ submatrix of $ (\widehat{\Sigma}_{AA})^{-1} $, and so on.
		
		First, we want to find a sufficient condition for (\ref{eq: Thm5-1 KKT3}) to hold.
		We divide $ \widehat \Sigma_{A_3, :} \widehat Z - \widehat M_{A_3} $ as follows:
		\begin{align*}
			&\widehat \Sigma_{A_3, :} \widehat Z - \widehat M_{A_3} \\
			&= \widehat \Sigma_{A_3 A_1} \widehat Z_{A_1} + \widehat\Sigma_{A_3 A_2} \widehat Z_{A_2}
			- \widehat M_{A_3} \\
			&= (\widehat \Sigma_{A_3 A_1}\tilde\Psi_{A_1} + \widehat \Sigma_{A_3A_2}\tilde\Psi_{A_2})
			- \lambda(\widehat \Sigma_{A_3 A_1}(\widehat \Sigma_{AA})^{-1}_{A_1 A_1} +
			\widehat \Sigma_{A_3 A_2}(\widehat \Sigma_{AA})^{-1}_{A_2 A_1}) \hat s_{A_1} \\
			&\quad - \lambda\eta(\widehat \Sigma_{A_3 A_1}(\widehat \Sigma_{AA})^{-1}_{A_1 A_2} +
			\widehat \Sigma_{A_3 A_2}(\widehat \Sigma_{AA})^{-1}_{A_2 A_2}) \hat s_{A_2}
			- \widehat M_{A_3} \\
			&= (\widehat \Sigma_{A_3A}\tilde\Psi_{A} - \widehat M_{A_3})
			- (\lambda \widehat \Sigma_{A_3A}(\widehat \Sigma_{AA})^{-1}_{AA_1}\hat s_{A_1} +
			\lambda\eta \widehat \Sigma_{A_3A}(\widehat \Sigma_{AA})^{-1}_{AA_2}\hat s_{A_2}) \\
			&\eqqcolon I_1 - I_2.
		\end{align*}
		
		We now find upper bounds for each of $\|I_1\|_{\infty, 2}$, and $\|I_2\|_{\infty, 2}$.
		We note that Lemma \ref{lemma: dividing}
		with $ (V, \Upsilon, \Delta) = (\Psi, \Sigma, M) $
		implies that
		$ \Sigma_{A^c A} (\Sigma_{AA})^{-1}M_A = M_{A^c} $.
		Since $ A^c = A_3 \cup A_4 $, it also holds that
		$ \Sigma_{A_3 A} (\Sigma_{AA})^{-1}M_A = M_{A_3} $.
		Then, we have
		\begin{align*}
			I_1
			&= \widehat \Sigma_{A_3 A} (\widehat \Sigma_{AA})^{-1} \widehat{M}_A - \widehat M_{A_3} \\
			&= (\widehat \Sigma_{A_3 A} (\widehat \Sigma_{AA})^{-1} - \Sigma_{A_3 A}(\Sigma_{AA})^{-1})
			(\widehat M_A - M_A + M_A) \\
			&\quad + \Sigma_{A_3A}(\Sigma_{AA})^{-1}(\widehat M_A - M_A) -(\widehat M_{A_3} - M_{A_3}).
		\end{align*}
        Note that we assume the event $\mathcal E_M(\epsilon) \cap \mathcal E_{\Sigma}(\epsilon) \cap \mathcal E_{ord}$
        for some $\epsilon > 0$.
		From Lemma \ref{lem:norms} and the first and the second parts of Lemma \ref{lem: concentrations}, we have
		\begin{equation} \label{eq:I-1}
		\|I_1\|_{\infty, 2} \le \frac{\phi \epsilon(1+\kappa)}{1 - \phi \epsilon}
		(\epsilon + \delta) + \kappa \epsilon + \epsilon
		= \frac{\epsilon(1+\kappa)(1+\phi\delta)}{1 - \phi\epsilon}.
		\end{equation}
		To bound $ \|I_2\|_{\infty, 2}$, observe that
		$$
		I_2 = \begin{bmatrix}
			\widehat \Sigma_{A_3 A} (\widehat \Sigma_{AA})^{-1}_{A A_1}
			& \widehat \Sigma_{A_3 A} (\widehat \Sigma_{AA})^{-1}_{A A_2}
		\end{bmatrix}
		\begin{bmatrix}
			\lambda \hat s_{A_1} \\
			\lambda\eta \hat s_{A_2}
		\end{bmatrix}
		= \widehat \Sigma_{A_3 A} (\widehat \Sigma_{AA})^{-1}
		\begin{bmatrix}
			\lambda \hat s_{A_1} \\
			\lambda\eta \hat s_{A_2}
		\end{bmatrix}.
		$$
		Again from Lemmas \ref{lem:norms} and \ref{lem: concentrations}, we have
		\begin{equation} \label{eq:I-2}
		\|I_2\|_{\infty, 2}
		\le \left(
		\frac{\phi\epsilon(1+\kappa)}{1 - \phi \epsilon} + \kappa
		\right) \lambda \eta
		= \frac{(\phi\epsilon+\kappa)\lambda\eta}{1 - \phi \epsilon}.
		\end{equation}
		By summing up, (\ref{eq:I-1}) and (\ref{eq:I-2}) yield
		$$
		\| \widehat \Sigma_{A_3, :} \widehat Z - \widehat M_{A_3} \|_{\infty, 2}
		\le
		\frac{\epsilon(1+\kappa)(1+\phi\delta)}{1 - \phi\epsilon}
		+ \frac{(\phi\epsilon+\kappa)\lambda\eta}{1 - \phi \epsilon}.
		$$
		However, note that for $ 0 \le \epsilon < \frac{1}{\phi} $,
		\begin{align*}
			&\frac{\epsilon(1+\kappa)(1+\phi\delta)}{1 - \phi\epsilon}
			+ \frac{(\phi\epsilon+\kappa)\lambda\eta}{1 - \phi \epsilon} \le \lambda \\
			&\Longleftrightarrow
			\epsilon(\lambda\phi(1 + \eta) + (1 + \kappa)(1 + \phi\delta)) \le \lambda(1-\eta\kappa).
		\end{align*}	
		Thus, for $ \eta < \frac{1}{\kappa} $, (\ref{eq: Thm5-1 KKT3}) is satisfied if
		$$
		\epsilon < \frac{\lambda(1 - \eta \kappa)}{\lambda\phi(1 + \eta) + (1 + \kappa)(1 + \phi\delta)}.
		$$
		
		By arguing the similar calculation steps conducted just before, one can easily see that
		(\ref{eq: Thm5-1 KKT4}) holds if
		$$
		\epsilon < \frac{\lambda\eta(1 - \kappa)}{2\phi\lambda\eta + (1 + \kappa)(1 + \phi\delta)}.
		$$
		However,
		$$
		\frac{\lambda(1 - \eta \kappa)}{\lambda\phi(1 + \eta) + (1 + \kappa)(1 + \phi\delta)} \le
		\frac{\lambda\eta(1 - \kappa)}{2\phi\lambda\eta + (1 + \kappa)(1 + \phi\delta)}
		$$
		since $ \eta \ge 1 $.
		So, for $ \lambda > 0 $ and $1 \le  \eta < 1/\kappa $, it is enough to satisfy that
		\begin{equation} \label{eq:eps-ub}
			\epsilon < \min \left\{
			\frac{1}{\phi},~
			\frac{\lambda(1 - \eta \kappa)}{\lambda\phi(1 + \eta) + (1 + \kappa)(1 + \phi\delta)}
			\right\}
		\end{equation}
		to guarantee $ \widehat Z $ to be a solution.

		Next, we will determine $ \epsilon $.
		One can check that if
		$
			\lambda < (1+\kappa)(\phi^{-1} + \delta)/2
		$
        and $\eta < \frac{1+\kappa^{-1}}{2}$,
		then
		\begin{align*}
			\frac{1}{4} \cdot \frac{\lambda(1-\kappa)}{(1+\kappa)(1+\phi\delta)} 
            < \frac{\lambda(1 - \eta \kappa)}{\lambda\phi(1 + \eta) + (1 + \kappa)(1 + \phi\delta)},
		\end{align*}
        since $1-\eta\kappa \ge (1-\kappa)/2$.
		Furthermore, if $ \lambda < 2\phi^{-1} $ then we get 
		\begin{align*}
   	    \frac{1}{4} \cdot \frac{\lambda(1-\kappa)}{(1+\kappa)(1+\phi\delta)} < \frac{1}{2\phi}.
		\end{align*}
		Thus, if we take $ \epsilon = \frac{1}{4} \frac{\lambda(1-\kappa)}{(1+\kappa)(1+\phi\delta)} $ when 
		$
			\lambda < \min( (1+\kappa)(\phi^{-1} + \delta)/2, 2\phi^{-1})
		$, and
		$ 1 \le \eta < (1 + \kappa^{-1})/2 $ then $ \widehat{Z} $ becomes the solution.
	\end{proof}
	
	\begin{proof}[Proof of 2.]
        For a $\epsilon > 0$, assume the event $ \mathcal E_{\Sigma}(\epsilon) \cap \mathcal E_{M}(\epsilon) \cap \mathcal E_{ord} $.
		From the KKT conditions, $ \widehat Z = [\widehat Z_{A_1}^\top,  0_{p-|A_1|}^\top]^\top $ becomes a solution of
		the optimization problem (8) if
		\begin{align}
			\widehat\Sigma_{A_1, :} \widehat Z - \widehat M_{A_1} + \lambda \hat s_{A_1} &= 0, \label{eq: Thm5 KKT1} \\
			\| \widehat \Sigma_{A_2, :} \widehat Z - \widehat M_{A_2}\|_{\infty, 2} &\le \lambda \eta , \label{eq: Thm5 KKT2} \\
			\|\widehat \Sigma_{A_3, :} \widehat Z - \widehat M_{A_3} \|_{\infty, 2} &\le \lambda, \label{eq: Thm5 KKT3} \\
			\|\widehat \Sigma_{A_4, :} \widehat Z - \widehat M_{A_4} \|_{\infty, 2} &\le \lambda \eta \label{eq: Thm5 KKT4}
		\end{align}	
		are satisfied.
		
		First, we want to find sufficient conditions that (\ref{eq: Thm5 KKT2}) holds.
		From (\ref{eq: Thm5 KKT1}), we have
		$
		\widehat Z_{A_1}
		= (\widehat \Sigma_{A_1 A_1})^{-1}(\widehat M_{A_1} - \lambda \hat{s}_{A_1})
		$
		since $ \widehat Z_{A_{1}^{c}} = 0 $.
		Note that $ \Sigma_{A_1 A_2} = 0 $.
		So, we have $ \kappa_2 = 0 $ and $ \phi_1 \le \phi $.
		From the first and the third parts of Lemma \ref{lem: concentrations}, we have
		\begin{align*}
			&\|\widehat \Sigma_{A_2 A_1} \widehat Z_{A_1} - \widehat M_{A_2}\|_{\infty, 2} \\
			&\le \|\widehat \Sigma_{A_2 A_1} (\widehat \Sigma_{A_1 A_1})^{-1}
			- \Sigma_{A_2 A_1} (\Sigma_{A_1 A_1})^{-1}\|_{\infty}
			(\|\widehat M_{A_1} - M_{A_1}\|_{\infty} + \|M_{A_1}\|_{\infty, 2} + \lambda) \\
			&\quad + \|\widehat{M}_{A_2} - M_{A_2}\|_{\infty} + \|M_{A_2}\|_{\infty, 2} \\
			&\le
			\frac{\phi_1\epsilon}{1 - \phi_1 \epsilon} (\epsilon + \delta_1 + \lambda)
			+ \epsilon + \delta_2 \\
			&\le \frac{\phi\epsilon}{1 - \phi \epsilon} (\epsilon + \delta_1 + \lambda)
			+ \epsilon + \delta_2,
		\end{align*}
		if $ \phi \epsilon < 1 $.
		Assuming so, it is enough to bound the last term of previous inequalities by
		$ \lambda\eta $ to guarantee (\ref{eq: Thm5 KKT2}).
		Then, we have
		\begin{align*}
			&\quad \frac{\phi\epsilon}{1 - \phi \epsilon} (\epsilon + \delta_1 + \lambda)
			+ \epsilon + \delta_2 \le \lambda \eta  \\
			&\Longleftrightarrow
			\epsilon(1 + \phi(\delta_1 - \delta_2 + \lambda + \lambda\eta))
			\le \lambda\eta - \delta_2.
		\end{align*}
		To make both sides of the inequality positive, we should choose sufficiently
		large $ \eta \ge 1 $ such that $\lambda\eta > \delta_2.$
		In summary, if $ \eta > \frac{\delta_2}{\lambda} $ and
		\begin{equation} \label{eq:proof-thm2-2-eps-bound1}
		\epsilon < \frac{\lambda\eta - \delta_2}
		{1 + \phi(\delta_1 - \delta_2 + \lambda + \lambda\eta)} \wedge \frac{1}{\phi}
		\end{equation}
		then (\ref{eq: Thm5 KKT2}) holds.
		
		In the next step, we characterize sufficient conditions that (\ref{eq: Thm5 KKT3}) holds.
		By repeating similar calculation, we get
		\begin{align*}
		&\widehat \Sigma_{A_3,:} \widehat Z - \widehat M_{A_3} \\
		&= (\widehat \Sigma_{A_3 A_1} (\widehat \Sigma_{A_1A_1})^{-1} - \Sigma_{A_3 A_1} (\Sigma_{A_1A_1})^{-1})
		(\widehat M_{A_1} -  M_{A_1} +  M_{A_1} - \lambda \hat s_{A_1})  \\
		&~+ \Sigma_{A_3 A_1} \Sigma_{A_1 A_1}^{-1}(\widehat M_{A_1}-M_{A_1}-\lambda\hat s_{A_1})
		+ M_{A_3} - \widehat M_{A_3}.
		\end{align*}
		Here, we use the fact that
		$$
			M_{A_3} = \Sigma_{A_3 A} (\Sigma_{AA})^{-1}M_{A}
			= \Sigma_{A_3 A_1} (\Sigma_{A_1A_1})^{-1}M_{A_1},
		$$
		since $ \Sigma_{A_2 A_2^c} =  0 $.
		Then,
		\begin{align*}
			\|\widehat \Sigma_{A_3,:} \widehat Z - \widehat M_{A_3}\|_{\infty, 2}
			\le \frac{\phi\epsilon(1+\kappa)}{1 - \phi \epsilon} (\epsilon + \delta_1 + \lambda) + \kappa(\epsilon + \lambda) + \epsilon.
		\end{align*}
		Here, we use the fact that $ \phi_1 \le \phi $ and $ \kappa_3 \le \kappa $.
		For $ \epsilon < \frac{1}{\phi} $,
		it holds that
		$$
			\frac{\phi\epsilon(1+\kappa)}{1 - \phi \epsilon} (\epsilon + \delta_1 + \lambda) + \kappa(\epsilon + \lambda) + \epsilon
			\le \lambda
		$$
		which is equivalent to
		$$
			\epsilon \le \frac{\lambda(1-\kappa)}{2\lambda \phi + (1+\phi\delta_1)(1+\kappa)}.
		$$
		In a similar manner, for $ \epsilon < \frac{1}{\phi} $,
		(\ref{eq: Thm5 KKT4}) holds when
		$$
			\epsilon \le \frac{\lambda(\eta-\kappa)}{\lambda(1+\eta) \phi + (1+\phi\delta_1)(1+\kappa)}.
		$$
		From the assumption $ \kappa < 1 $ and the choice of $ \eta \ge 1 $, we get
		\begin{equation} \label{eq:proof-thm2-2-eps-bound2}
			\frac{\lambda(1-\kappa)}{2\lambda \phi + (1+\phi\delta_1)(1+\kappa)}
			\le
			\frac{\lambda(\eta-\kappa)}{\lambda(1+\eta) \phi + (1+\phi\delta_1)(1+\kappa)}.
		\end{equation}
		
		So, for $ \lambda > 0 $ and $ \eta > 1 \vee \frac{\delta_2}{\lambda} $,
		it is enough to satisfy that
		$$
			\epsilon < \min\left\{
			\frac{1}{\phi},~
			\frac{\lambda\eta - \delta_2}
			{1 + \phi(\delta_1 - \delta_2 + \lambda + \lambda\eta)},~
			\frac{\lambda(1-\kappa)}{2\lambda \phi + (1+\phi\delta_1)(1+\kappa)}
			\right\}
		$$
		to guarantee $ \widehat Z $ to be a solution.
		The upper bound of $\epsilon$ comes from
		(\ref{eq:proof-thm2-2-eps-bound1}) and (\ref{eq:proof-thm2-2-eps-bound2}).

		Finally, we will determine $ \epsilon $ as in the previous proof.
		If $ \lambda < \frac{(1+\phi\delta_1)(1+\kappa)}{2\phi}$, then 
		$$
		 	\frac{1}{2}\cdot \frac{\lambda(1-\kappa)}{(1+\phi\delta_1)(1+\kappa)}
			< \frac{\lambda(1-\kappa)}{2\lambda \phi + (1+\phi\delta_1)(1+\kappa)}.
		$$
		On the other hand, if $\lambda < \frac{1}{\phi}$ and $ \beta < \lambda\eta - \delta_2 < \beta'$ for some $\beta, \beta' > 0$, it holds that
		$$
		\frac{\beta}{2+\phi(\beta'+\delta_1)} < \frac{\lambda\eta - \delta_2}
			{1 + \phi(\delta_1 - \delta_2 + \lambda + \lambda\eta)}.
		$$
		Thus, if 
		$$
			\lambda < \frac{2\beta(1+\phi\delta_1)(1+\kappa)}{(2+\phi(\beta' + \delta_1))(1-\kappa)}
		$$
		then 
		$$
			\frac{1}{2}\cdot \frac{\lambda(1-\kappa)}{(1+\phi\delta_1)(1+\kappa)} <  \frac{\lambda\eta - \delta_2}
			{1 + \phi(\delta_1 - \delta_2 + \lambda + \lambda\eta)}.
		$$
        If we take $ \epsilon = \frac{1}{2}\cdot \frac{\lambda(1-\kappa)}{(1+\phi\delta_1)(1+\kappa)} $ when $ \lambda < \min(\frac{2\beta(1+\phi\delta_1)(1+\kappa)}{(2+\phi(\beta' + \delta_1))(1-\kappa)}, \frac{1}{2\phi}) $ and $ \beta < \lambda\eta - \delta_2 < \beta'$ then $ \widehat Z $ becomes a solution.
	\end{proof}

	\subsubsection{Proof of Theorem 3}
        In the proof, we assume the event $ \mathcal E_{\Sigma}(\epsilon) \cap \mathcal E_{M}(\epsilon) \cap \mathcal E_{ord} $, as in the proof of Theorem 2. 
	\begin{proof}[Proof of 1]
	From the proof of the first part of Theorem 2, we have
		$ \widehat Z_A = (\widehat \Sigma_{AA})^{-1}(\widehat M_A - \Lambda_A \hat s_A) $ and
		$\widehat Z_{A^c} =0$,
		where $ \Lambda_A = \mbox{diag}(\lambda I_{A_1}, \lambda\eta I_{A_2}) $.
		Also, from Lemma \ref{lemma: dividing},
		$\Psi_A = (\Sigma_{AA})^{-1}M_A$.
		Then, we have
		\begin{align*}
			\widehat Z_A - \Psi_A
			&= (\widehat \Sigma_{AA})^{-1}(\widehat M_A - \Lambda_A \hat s_A) - (\Sigma_{AA})^{-1}M_A \\
			&= \left( (\widehat \Sigma_{AA})^{-1} - (\Sigma_{AA})^{-1} \right)
			(\widehat M_A - M_A + M_A - \Lambda_A \hat s_A) \\
			&~ +(\Sigma_{AA})^{-1} (\widehat M_A - M_A - \Lambda_A \hat s_A).
		\end{align*}
		Lemma \ref{lem: concentrations} implies that
		\begin{align*}
			\|\widehat Z_A - \Psi_A \|_{\infty, 2}
			&\le \frac{\epsilon \phi^2}{1 - \phi \epsilon} (\epsilon + \delta + \lambda \eta)
			+ \phi(\epsilon + \lambda\eta) \\
			&= \frac{\phi \epsilon}{1-\phi\epsilon} (1+\phi\delta) + \frac{\phi\lambda\eta}{1-\phi\epsilon}.
		\end{align*}
		If
          \begin{equation}\label{eq:thm3_eps_bd}
              \epsilon < \frac{1}{2\phi} \wedge \frac{2\lambda\eta}{1+\phi\delta},
          \end{equation}
        then it follows that
		$$
		\frac{\phi \epsilon}{1-\phi\epsilon} (1+\phi\delta) + \frac{\phi\lambda\eta}
		{1-\phi\epsilon}
		< 2\phi\frac{2\lambda\eta}{1+\phi\delta}(1+\phi\delta)  + 2\phi\lambda\eta
		= 6\phi\lambda\eta.
		$$
        From the proof of the first part of Theorem 2, for $\lambda < \min( (1+\kappa)(\phi^{-1} + \delta)/2, 2\phi^{-1})$ and $ 1 \le \eta < (1 + \kappa^{-1})/2 $, $$ \epsilon = \frac{1}{4} \frac{\lambda(1-\kappa)}{(1+\kappa)(1+\phi\delta)} $$
        satisfies \eqref{eq:thm3_eps_bd}. This completes the proof.
	\end{proof}

	\begin{proof}[Proof of 2]
		From the proof of the second part of Theorem 2, we have
		$$\widehat Z_{A_1} = (\widehat \Sigma_{A_1 A_1})^{-1}(\widehat M_{A_1} - \lambda
		\hat s_{A_1})$$
		and $\Psi_{A_1} = (\Sigma_{A_1 A_1})^{-1}M_{A_1}$ since we assume that
		$\Sigma_{A_1 A_2} = 0$.
		Then, we have
		\begin{align*}
			\widehat Z_{A_1} - \Psi_{A_1}
			&= \left( (\widehat \Sigma_{A_1A_1})^{-1} - (\Sigma_{A_1A_1})^{-1} \right)
			(\widehat M_{A_1} - M_{A_1} + M_{A_1} - \lambda \hat s_{A_1}) \\
			&~ +(\Sigma_{A_1A_1})^{-1} (\widehat M_{A_1} - M_{A_1} - \lambda \hat s_{A_1}).
		\end{align*}
		So, by Lemma \ref{lem: concentrations} and $ \phi_1 \le \phi $, we have
		\begin{align*}
			\|\widehat Z_{A_1} - \Psi_{A_1}\|_{\infty, 2}
			\le \frac{\phi\epsilon}{1-\phi\epsilon}(1+\phi \delta_1)
			+ \frac{\phi \lambda}{1 - \phi \epsilon}.
		\end{align*}
		If $\epsilon < \frac{1}{2\phi} \wedge \frac{2\lambda}{1+\phi\delta_1}$, then
		$$
		\frac{\phi\epsilon}{1-\phi\epsilon}(1+\phi \delta_1)
		+ \frac{\phi \lambda}{1 - \phi \epsilon}
		< 4\phi\lambda + 2\phi\lambda = 6\phi\lambda.
		$$
		As in the previous proof, the choice of $ \epsilon $ in the proof of the second part of Theorem 2 satisfies the condition, and the proof is done.
	\end{proof}
	
	\subsubsection{Proof of Corollaries}
        \textbf{Proof of Corollaries 4 and 5}
	\begin{proof}
		We only prove the Corollary 4. 
	    To make sense that \eqref{eq: proof-thm1-4} in the proof of Theorem 1 with varying $\delta_{\min}$, $N$ must grows fast to satisfy $\sqrt N \wedge \delta_{\min} \to \infty$ as $N \to \infty$.
	    Actually, this holds from the conditions (AC1) and (AC3). 
		Thus, it is enough to show that (AC1)-(AC3) imply $\gamma \to 0$ as $N \to \infty$.
		The following simple lemma is the key ingredient.

		\begin{lem}\label{lem: asymp}
		Suppose that $a_n, b_n \to \infty$. If $\frac{\log a_n}{b_n} \to 0$ then
		$$a_n\exp(-Cb_n) \to 0.$$
		\end{lem}

		From (AC1) and (AC2),
		$
		\frac{\log(pd)}{N\lambda_N^2/d^2} \to 0.
		$
		and by Lemma \ref{lem: asymp}, we have
		$
		pd\exp(-CN\lambda_N^2/d^2) \to 0
		$
		which shows that the first term of $ \gamma$ tends to $ 0 $.
		Similarly, the second and third terms of $ \gamma $ converge to $ 0 $ under the asymptotic conditions.

	\end{proof}

        \noindent
        \textbf{Proof of Corollary 6}
        \begin{proof}
            The first part of the corollary can be deduced from the results of \cite{gaynanova2016sparseLDA} and \cite{mai2019multiclass} for consistent variable selection of sparse multiclass LDA basis and then applying Theorem 1 for the screening procedure.

            For the second part, we use the result of \cite{gaynanova2015optimal} for the exact variable selection property of MGSDA. Under the asymptotic conditions of the corollary, $\widehat{D}(\widehat Z_{\lambda}) = J_{disc}$ holds with a high probability, and then applying Theorem 1 for the screening procedure completes the proof.
        \end{proof}

\section{Block-coordinate descent algorithm for proposed methods} \label{appendix:alg}
The convex minimization problem (6) of the main paper can be efficiently  solved by a block-coordinate descent algorithm, which is guaranteed to converge to the global minimum \citep{tseng1993desc}.
Write $ Z = [z_1, \dots, z_{K-1}]$, and
$ \widehat M = [\hat m_1, \dots, \hat m_d] = [\widetilde M_1, \dots, \widetilde M_p]^\top$
and recall that $ \widetilde Z_j $ is the $ j $th row vector of $ Z $.
We write $z_{ij}$, $\hat m_{ij}$ and $ \hat \sigma_{ij}$ for the  $ (i,j) $th element of
$ Z$, $\widehat M $ and $ \widehat \Sigma $.
Then the objective function of (6) becomes
\begin{align*}
    &\mbox{trace}\left(
    \frac{1}{2} Z^\top \widehat \Sigma Z - Z^\top \widehat M
    \right)	+ \lambda \sum_{j=1}^{p} \eta^{1-w_i} \|\widetilde Z_j\|_2 \\
    &= \sum_{r=1}^{K-1} \left(
    \frac{1}{2} z_r^\top\widehat \Sigma z_r - z_r^\top \hat m_r
    \right)	+ \lambda \sum_{j=1}^{p} \eta^{1-w_j} \|\widetilde Z_j\|_2\\
    &= \sum_{r=1}^{K-1} \left(
    \frac{1}{2} \sum_{i=1}^{p} \sum_{k=1}^{p} z_{ir} \hat\sigma_{ik} z_{kr}
    - \sum_{i=1}^{p} z_{ir} \hat m_{ir}
    \right) + \lambda \sum_{j=1}^{p} \eta^{1-w_j} \|\widetilde Z_j\|_2\\
    &= \frac{1}{2} \sum_{i=1}^{p} \sum_{k=1}^{p} \hat\sigma_{ik} \widetilde Z_i^T \widetilde Z_k
    - \sum_{i=1}^{p} \widetilde Z_i^\top \widetilde M_i
    + \lambda \sum_{j=1}^{p} \eta^{1-w_j} \|\widetilde Z_j\|_2 \\
    &\eqqcolon f(\widetilde Z_1, \dots, \widetilde Z_p).
\end{align*}
Thus, the convex problem (6) is block separable with respect to the rows of $ Z $.
Our algorithm updates rows of $ Z $ iteratively in a cyclic manner.
At the $ (t+1) $th iteration, the algorithm updates $ \widetilde Z_i^{(t+1)} $ by
\begin{equation*}
    \widetilde Z_i^{(t+1)} = \underset{\widetilde Z \in \mathbb{R}^{K-1}}{\arg\min}
    f(\widetilde Z_1^{(t+1)}, \dots, \widetilde Z_{i-1}^{(t+1)}, \widetilde Z,
    \widetilde Z_{i+1}^{(t)}, \dots, \widetilde Z_{p}^{(t)}),
\end{equation*}
for $ i = 1, \dots, p $ consecutively.
More precisely, $ \widetilde Z_i^{(t+1)} $ is the solution of
\begin{equation}\label{eq:subproblem}
    \min_{\widetilde Z \in \mathbb{R}^{K-1}} \left\{
    \frac{1}{2} \hat\sigma_{ii} \widetilde Z^\top\widetilde Z - a_i^\top \widetilde Z
    + \lambda \eta^{1-w_i} \|\widetilde{Z}\|_2
    \right\}
\end{equation}
where $ a_i = \widetilde M_i- \sum_{j \neq i} \hat\sigma_{ij} \widetilde Z_j^{(t')} $
with $ t' = t+1 $ if $ j < i $ and $ t' = t $ if $ j > i $.
The solution of (\ref{eq:subproblem}) is given by
\[
\widetilde Z_i^{(t+1)} = \frac{1}{\hat\sigma_{ii}}
\left(
1 - \frac{\lambda \eta^{1-w_i}}{\|a_i\|_2}
\right)_+ a_i
\]
where $ (x)_+ = \max(x, 0) $.

    \begin{remark*}
        Note that both of our proposed methods, SOBL and OSBL, as well as existing approaches of SLDA and FWOC, use the block-coordinate descent algorithm to obtain their sparse basis. Thus, the computing times required for each method depend on the parameter tuning procedures (as well as the computation of ordinal weights). In our numerical studies, SOBL (including the two-step tuning procedure) is on average two times slower than SLDA (including the cross-validation), but is slightly faster than FWOC. 
        The computation times for OSBL are comparable to SLDA.
    \end{remark*}

\section{Simulation settings for Section 3.5} 
In Section 3.5 of the main paper, we present a numerical comparison of the ordinal weights, based on Spearman's rank correlation, Kendall's $ \tau $, the $t$-test-based trend test, and the two-step procedure.
To generate a simulation data, we consider $ K=4 $ groups with group means
$$  
[\mu_1, \mu_2, \mu_3, \mu_4] =
0.5 \times
\left[\begin{array}{ccccc:ccccc:c}
	0 & 0 & 0 & 0 & 0 & 0  & 0 & 6 & 6 & 2 & \mathbf 0^{\top}_{p-10}\\
	2 & 1 & 1 & 3 & 2 & -2 & 4 & 3 & 0 & 0 & \mathbf 0^{\top}_{p-10}\\
	4 & 3 & 4 & 5 & 3 & -4 & 2 & 0 & 3 & 6 & \mathbf 0^{\top}_{p-10}\\
	6 & 6 & 6 & 6 & 6 & 2  & 6 & 5 & 5 & 4 & \mathbf 0^{\top}_{p-10}
\end{array}\right]^\top
$$
and common variance $ \Sigma_{ij} = 0.6^{|i-j|} $, where $ p = 200 $.
We assume Gaussian model, so $ X\:|\:Y = j \sim N(\mu_j, \Sigma) $.
We generate $ 50 $ random samples per group and calculate ordinal weights. The ordinal weights, averaged over 50 repetitions, are shown in
Figure 2 of the main paper.

\section{Biological findings in real data analysis}

Here, we briefly investigate the variables selected by our methods in the glioma and ALL  data analyses. 

\begin{figure}
	\begin{center}
		\includegraphics[width=\textwidth]{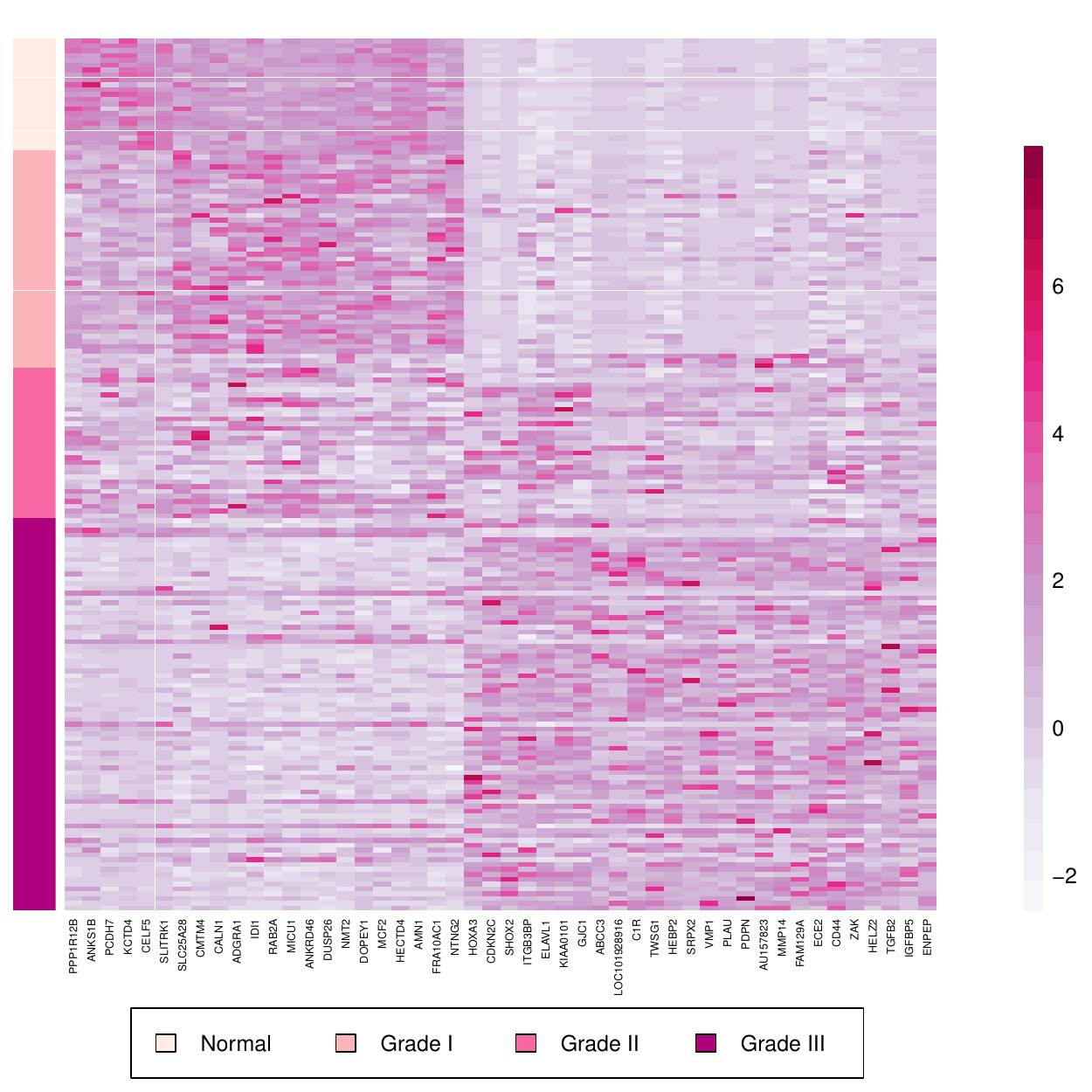}
		\caption{Heatmap of selected variables of Glioma dataset. Each column corresponds to one of the selected genes, and each row denotes gene expression levels of an observation.}
		\label{fig:heatmap-glioma}
	\end{center}
\end{figure}
\begin{figure}
	\begin{center}
		\includegraphics[width=\textwidth]{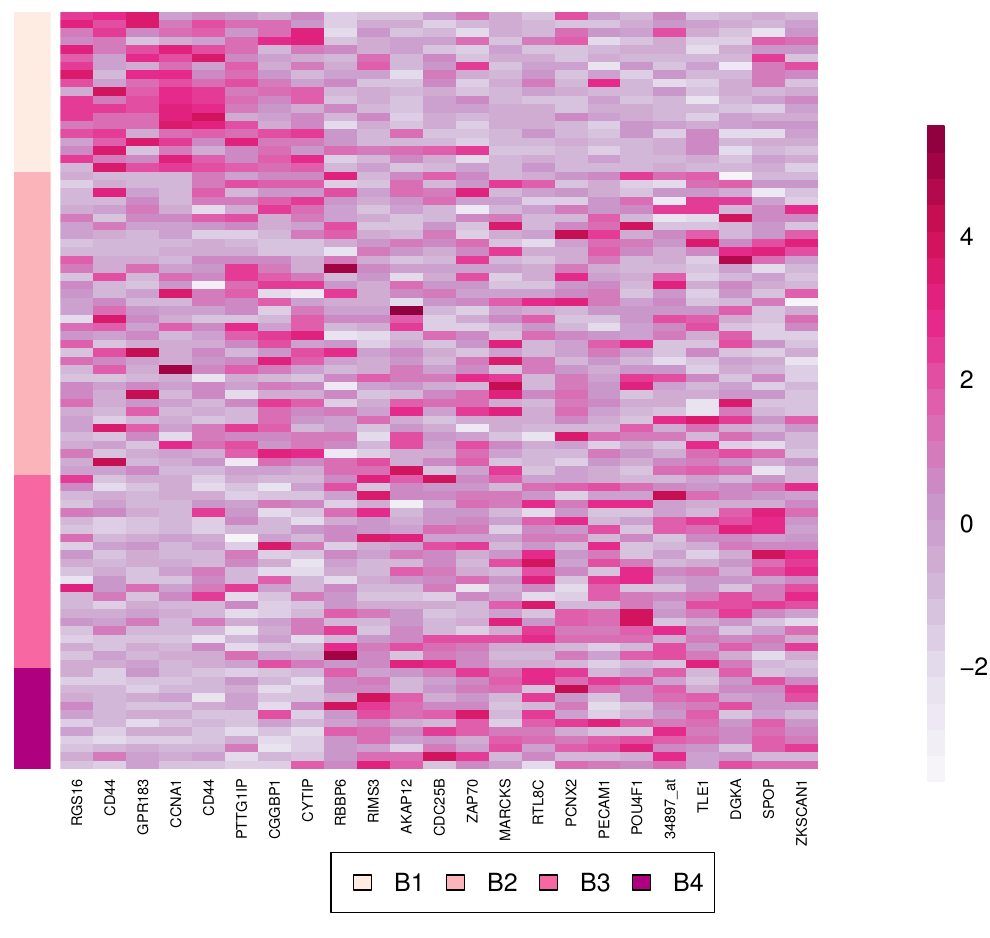}
		\caption{Heatmap of selected variables of ALL dataset. Each column corresponds to one of the selected genes, and each row denotes gene expression levels of an observation.}
		\label{fig:heatmap-all}
	\end{center}
\end{figure}

Figures \ref{fig:heatmap-glioma} and \ref{fig:heatmap-all} show the heatmaps of the gene expression level of the selected variable for each dataset. For Glioma data, two clusters of genes exist. Genes from PPP1R12B to NTNG2 are one cluster, and the other cluster consists of genes from HOXA3 to ENPEP. Genes in the first cluster have higher expression levels in Normal and Grade I patients, which decreases as the group varies from Normal to Grade III. Conversely, genes in the second cluster have higher expression levels at Grade II and Grade III and have low levels at Normal and Grade I. 
For ALL dataset, the pattern of the heatmap is not as clear as the case of the glioma dataset, but it has a similar tendency. For genes RGS16 to CYTIP, expression levels decrease as the label increases (from $B_1$ to $B_4$). For the rest of the genes, RBBP6 to ZKSCAN1, expression levels are increasing as the group label increases.

To decide on important variables among the selected variables, we plot Kendall's rank correlation between each variable and gene expression level; see Figures \ref{fig:kendall-glioma} and \ref{fig:kendall-all}. Though the all selected variables have significantly high correlation value, we only focus on few variables here. Below discussion shows that our proposed methods well finds biologically important variables and may find potential biomarkers which does not be investigated.

\begin{figure}
	\begin{center}
		\includegraphics[width=\textwidth]{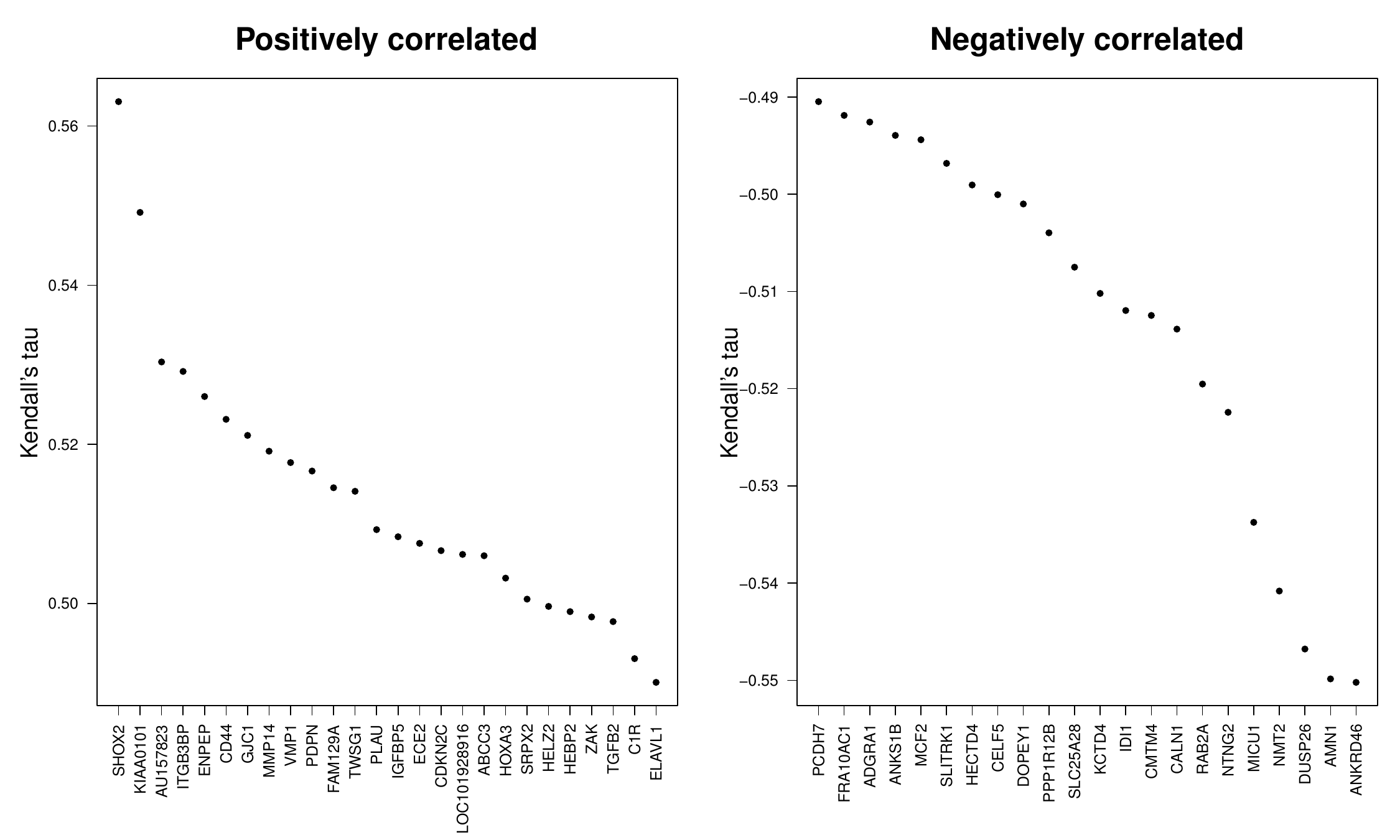}
		\caption{Kendall's rank correlation plots of selected variables in Glioma data analysis.}
		\label{fig:kendall-glioma}
	\end{center}
\end{figure}
\begin{figure}
	\begin{center}
		\includegraphics[width=\textwidth]{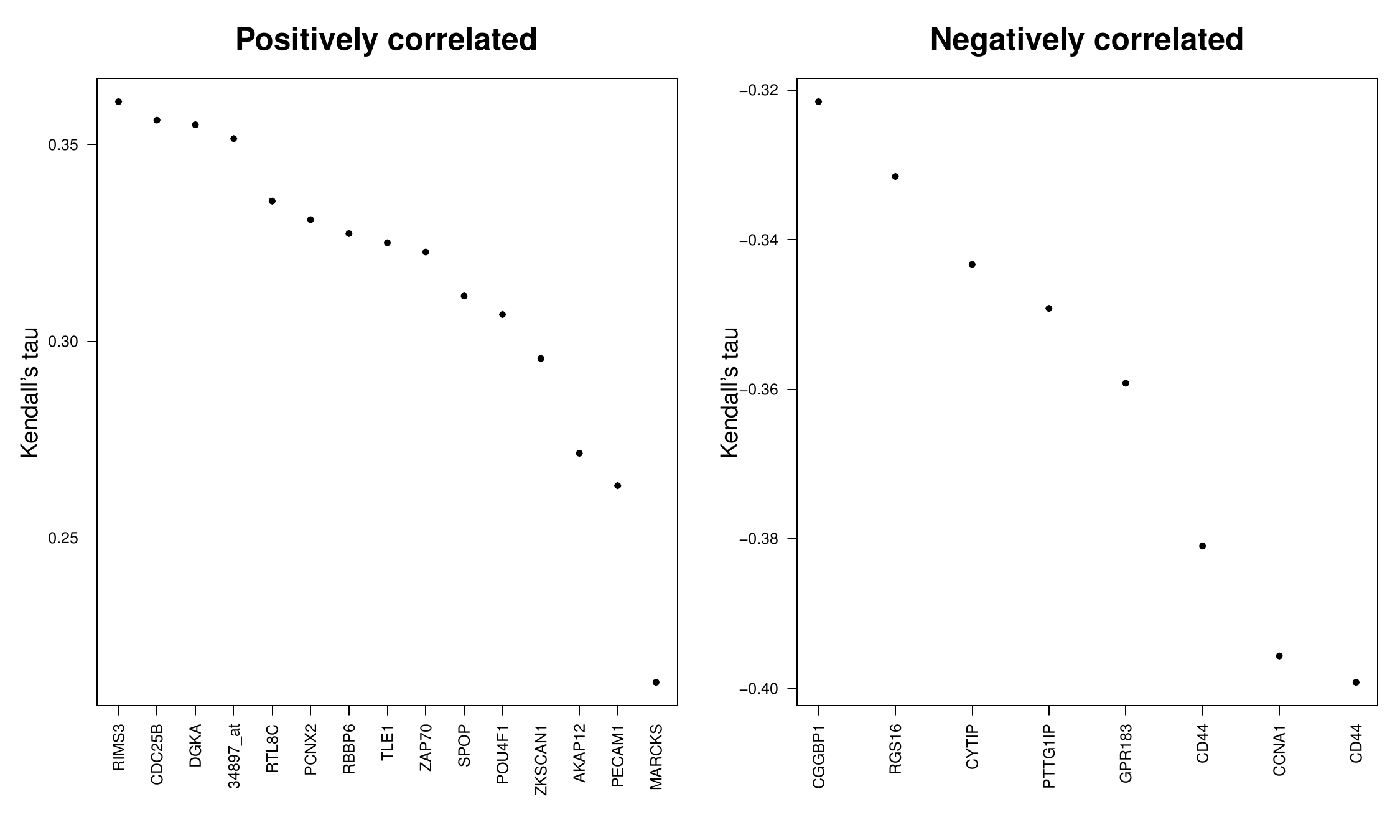}
		\caption{Kendall's rank correlation plots of selected variables in ALL data analysis.}
		\label{fig:kendall-all}
	\end{center}
\end{figure}

\subsection{Selected genes from Glioma data analysis}
In Glioma data analysis, the most positively correlated variable is SHOX2. (See the left panel of Figure \ref{fig:kendall-glioma}.) This gene was suggested as a potential prognostic indicator for grade II and III diffuse gliomas \citep{zhang2016shox2}. We remark that SHOX2 is recognized as a meaningful gene in the progression of various cancers. For example, \cite{yang2023shox2} experimentally studied how SHOX2 is associated with prostate cancer. The second most positively correlated gene KIAA0101 also has a known relationship between Glioma with experimental evidence \citep{wang2022kiaa0101}. Also, ITGB3BP \citep{liu2022itgb3bp}, ENPEP \citep{oudenaarden2022upregulated, xu2022whole}, and CD44 \citep{du2022association} are also known to be important genes related to the prognosis of Glioma.

Next, we examine some negatively correlated genes. It is worth noting that the two most negatively correlated genes, ANKRD46 and AMN1, do not have a known relationship to Glioma to the best of our knowledge. Elucidating relations between these undiscovered genes and Glioma may be helpful to progress understanding of Glioma.

There are also known important negatively correlated genes that may be involved with Glioma. DUSP26 is potentially guessed as an oncogene \citep{navis2010protein}, and NMT2 was used to target an inhibitor as a radiosensitizer in glioblastoma \citep{dinakaran2023radiosensitization}. As discoverd by \cite{aulestia2018quiescence}, MICU1, which is the modulator of the mitochondrial $\mbox{Ca}^{2+}$ uniporter, is related to the status of proliferating glioblastoma stem-like cells.

\subsection{Selected genes from ALL data analysis}
In ALL data analysis, RIMS3 is the most positively correlated variables as illustrated in Figure \ref{fig:kendall-all}. In the literature of leukemia study, RIMS3 was already discovered that it is highly upregulated with acute leukemia \citep{hicks2016molecular, malmberg2019accurate}. However, we cannot find some experimental results for RIMS3.
CDC25B is known to be involved in acute myeloid leukemia (ALM), but it is not clear that CDC25B is also related to ALL \citep{nakamura2011transcriptional}.
Similarly, DGKA is also known as related to ALM \citep{tan2023ritanserin} but not known for ALL to the best of our knowledge.
It is interesting that the selected gene of probe ID `34987\_at', which has a high correlation, does not have any gene symbol in Affymetrix Human Genome U95 Set annotation data \citep{hgu95av2}. Investigating the gene `34987\_at' may provide helpful direction to understanding ALL.

Next, we consider negatively correlated variables. The most negatively correlated variable is CD44, which is importantly involved in ALL \citep{garcia2018notch1, piya2022targeting}.
CCNA1 is known to be related to AML \citep{leung2020evaluation, yang2021histone} and \cite{niss2017ebi2} experimentally showed that GPR183 is involved in ALL.

\end{document}